%% file: tpds2016.tex
\documentclass[10pt,journal,compsoc]{IEEEtran}
\pdfoutput=1
\usepackage{ifpdf}

\ifCLASSOPTIONcompsoc
\else
  \usepackage{cite}
\fi
\ifCLASSINFOpdf
\usepackage[pdftex]{graphicx}
\else
\usepackage[dvips]{graphicx}
\fi
\usepackage[cmex10]{amsmath}
\usepackage{url}

\usepackage{tikz}
\usetikzlibrary{calc} 

\usepackage{amsthm,amssymb,amsmath}
\usepackage{alltt}
\usepackage{listings}
\usepackage{multirow}

\newtheorem{proposition}{Proposition}

\newtheorem{proofpart}{Part}
\let\oldproofpart\proofpart
\renewcommand{\proofpart}{\oldproofpart\normalfont}
\makeatletter
\@addtoreset{proofpart}{theorem}
\makeatother

\usepackage{thmtools, thm-restate}

\usepackage{etoolbox}
\providetoggle{sketches}
\settoggle{sketches}{false}

\usepackage{xspace}
\usepackage{balance}

\usepackage{hhline}
\usepackage{rotating}

\usepackage{extarrows}

\usepackage{algorithm}
\usepackage{algorithmicx,algpseudocode}
\usepackage{multicol}
\usepackage{lipsum}
\usepackage{url}
\usepackage{listings}
\usepackage{ifthen}
\newboolean{techreport}
\setboolean{techreport}{false}
\usepackage{balance}

\begin{document}
%


\title{Hybrid Transactional Replication: State-Machine and Deferred-Update Replication Combined}



%
%
%
%

\author{Tadeusz~Kobus, 
        Maciej~Kokoci\'nski
        Pawe{\l}~T.~Wojciechowski
\IEEEcompsocitemizethanks{\IEEEcompsocthanksitem 
The authors are with the Institute of Computing Science, Poznan University of Technology, 
60-965 Pozna\'n, Poland. \protect\\
E-mail: \{Tadeusz.Kobus,Maciej.Kokocinski,Pawel.T.Wojciechowski\}{\hfil\break}@cs.put.edu.pl 
\IEEEcompsocthanksitem 
This work was funded from National Science Centre (NCN) funds granted by decision
No. DEC-2011/01/N/ST6/06762 and from Foundation for Polish Science (FNP) funds granted by decision No. 103/UD/SKILLS/2014.

%
 }}

\IEEEcompsoctitleabstractindextext{%
\begin{abstract}
\input{abstract.tex}

\end{abstract}

\begin{IEEEkeywords}
state machine replication; transactional replication; deferred update; distributed transactional memory
\end{IEEEkeywords}}

\maketitle

\IEEEdisplaynontitleabstractindextext

%
\IEEEpeerreviewmaketitle



\input{pseudocode_settings.tex}

\input{introduction.tex}
\input{related_work.tex}

\input{context.tex}

\input{htr.tex}

\input{evaluation.tex}

\input{conclusions.tex}

\ifCLASSOPTIONcompsoc
  \section*{Acknowledgments}
\else
  \section*{Acknowledgment}
\fi

We thank the
Pozna\'n Supercomputing and Networking Center (PSNC)
for providing computing resources.

%






\ifCLASSOPTIONcaptionsoff
  \newpage
\fi



\bibliographystyle{IEEEtran}
\bibliography{bibliography}

\balance

\begin{IEEEbiographynophoto}{Tadeusz Kobus}
is currently pursuing a Ph.D. degree and working as a Research Assistant 
in the Institute of Computing Science, Poznan University of Technology, 
Poland, where he also received B.S. and M.S. degrees in Computer 
Science, in 2009 and 2010 respectively. In the summer of 2014, he was 
an intern at IBM T. J. Watson Research Center. His research interests 
include fault tolerant distributed algorithms, transactional memory, and 
group communication systems.
\end{IEEEbiographynophoto}


\begin{IEEEbiographynophoto}{Maciej Kokoci\'nski}
is currently pursuing a Ph.D. degree and working as a Research Assistant 
in the Institute of Computing Science, Poznan University of Technology, 
Poland, where he also received B.S. and M.S. degrees in Computer 
Science, in 2009 and 2010 respectively. He was a summer intern at Microsoft 
in Redmond. His research interests include theory of distributed systems 
and transactional memory.
\end{IEEEbiographynophoto}

\begin{IEEEbiographynophoto}{Pawe{\l} T. Wojciechowski}
received his Ph.D. degree in Computer Science from the University of Cambridge
in 2000. He was a postdoctoral researcher in the School of Computer and 
Communication Sciences 
at \'Ecole Polytechnique F\'ed\'erale de Lausanne (EPFL), Switzerland, 
from 2001 to 2005. He holds a Habilitation degree from Poznan University 
of Technology, Poland, where he is currently an Assistant Professor 
in the Institute of Computing Science. 
He has led many research projects and coauthored dozens of papers. 
His research interests span topics in concurrency, distributed computing, and
programming languages.
\end{IEEEbiographynophoto}

\balance




\clearpage

\appendix 
\appendices
\input{appendix.tex}

\end{document}

%% file: abstract.tex
We propose \emph{Hybrid Transactional Replication (HTR)} a novel replication 
scheme for highly dependable services. It combines two schemes: a transaction 
is executed either optimistically by only one service replica in the deferred
update mode (DU), or deterministically by all replicas in the state machine 
mode (SM); the choice is made by an oracle. The DU mode allows for parallelism 
and thus takes advantage of multicore hardware. In contrast to DU, the SM mode 
guarantees abort-free execution, so it is suitable for irrevocable operations 
and transactions generating high contention. For expressiveness, transactions 
can be discarded or retried on demand. We formally prove that the higher 
flexibility of the scheme does not come at the cost of weaker guarantees for 
clients: HTR satisfies strong consistency guarantees akin to those provided by 
other popular transactional replication schemes such as Deferred Update 
Replication. We developed HTR-enabled Paxos STM, an object-based distributed 
transactional memory system, and evaluated it thoroughly under various 
workloads.
We show the benefits of using a novel oracle, which relies on machine learning 
techniques for automatic adaptation to changing conditions. In our tests, the 
ML-based oracle provides up to 50\% improvement in throughput when compared 
to the system running with DU-only or SM-only oracles. Our approach is inspired 
by a well known algorithm used in the context of the multi-armed bandit 
problem.

%% file: pseudocode_settings.tex
\newcommand{\adel}{\mathit{adel}}
\newcommand{\args}{\mathit{args}}
\newcommand{\clock}{\mathit{clock}}
\newcommand{\code}{\mathit{prog}}
\newcommand{\eend}{\mathit{end}}
\newcommand{\failure}{\mathit{failure}}
\newcommand{\id}{\mathit{id}}
\newcommand{\LC}{\mathit{LC}}
\newcommand{\Log}{\mathit{Log}}
\newcommand{\mode}{\mathit{mode}}
\newcommand{\nnull}{\mathit{null}}
\newcommand{\obj}{\mathit{obj}}
\newcommand{\oid}{\mathit{oid}}
\newcommand{\outcome}{\mathit{outcome}}
\newcommand{\prog}{\mathit{prog}}
\newcommand{\readset}{\mathit{readset}}
\newcommand{\res}{\mathit{res}}
\newcommand{\start}{\mathit{start}}
\newcommand{\stats}{\mathit{stats}}
\newcommand{\success}{\mathit{success}}
\newcommand{\txoracle}{\mathit{TransactionOracle}}
\newcommand{\update}{\mathit{update}}
\newcommand{\updates}{\mathit{updates}}
\newcommand{\vvalue}{\mathit{value}}
\newcommand{\version}{\mathit{version}}
\newcommand{\writeset}{\mathit{writeset}}

\algnewcommand\algorithmicupon{\textbf{upon}}
\algdef{SE}[UPON]{Upon}{EndUpon}                                   
   [2]{\algorithmicupon\ \textproc{#1}\ifthenelse{\equal{#2}{}}{}{(#2)}}%
   {\algorithmicend\ \algorithmicupon}%

\algnewcommand\abcastdesc{abcast}
\algnewcommand\Abcast{\abcastdesc{} }

\algnewcommand\tobcastdesc{\textproc{TO-Broadcast}}
\algnewcommand\tobcast{\tobcastdesc{} }

\algnewcommand\todeliverdesc{\textproc{TO-Deliver}}
\algnewcommand\todeliver{\todeliverdesc{} }

\algnewcommand\rbcastdesc{\textproc{R-Broadcast}{}}
\algnewcommand\rbcast{\rbcastdesc{} }

\algnewcommand\rdeliverdesc{\textproc{R-Deliver}}
\algnewcommand\rdeliver{\rdeliverdesc{}}

\algnewcommand\raisedesc{\textbf{raise}}
\algnewcommand\Raise{\raisedesc{} }

\algnewcommand\gotodesc{\textbf{goto}}
\algnewcommand\Goto{\gotodesc{} }

\algnewcommand\endrequestdesc{return $(r.\id, \LC, \res_q)$ to client $c$}
\algnewcommand\EndReq{\endrequestdesc{} }%

\algnewcommand\endrequestbasicdesc{return $(r.\id, \res_q)$ to client $c$}
\algnewcommand\EndReqBasic{\endrequestbasicdesc{} }%

\algnewcommand\endtransactiondesc{stop executing transaction $t$}
\algnewcommand\EndTx{\endtransactiondesc{} }%

\algnewcommand\endtransactionreturnrbdesc{stop executing $r.\prog$ and 
return to \Call{init}{}}
\algnewcommand\EndTxR{\endtransactionreturnrbdesc{} }%

\algnewcommand\endtransactionreturtoinitdesc{return to \Call{init}{}}
\algnewcommand\EndTxC{\endtransactionreturtoinitdesc{} }%

\algnewcommand\endtransactioncdesc{stop executing $r.\prog$}
\algnewcommand\EndTxCode{\endtransactioncdesc{} }%

\algnewcommand\endsmtransactionreturnrbdesc{stop executing $r.\prog$ and 
return to \Call{TO-Deliver}{}}
\algnewcommand\EndSmTx{\endsmtransactionreturnrbdesc{} }%

\algnewcommand\lockstartdesc{\textbf{lock \{}}
\algnewcommand\LockStart{\lockstartdesc{} }

\algnewcommand\lockenddesc{\textbf{\}}}
\algnewcommand\LockEnd{\lockenddesc{} }

\algnewcommand\algindentdesc{\hspace{2.8em}}
\algnewcommand\AlgIndent{\algindentdesc{} }

\algnewcommand\algindentsmalldesc{\hspace{1.3em}}
\algnewcommand\AlgIndentSmall{\algindentsmalldesc{} }

\algnewcommand\algindentindentdesc{\hspace{4.0em}}
\algnewcommand\AlgIndentIndent{\algindentindentdesc{} }

\algnewcommand\algexecutedesc{execute $r.\prog$ with $r.\args$}
\algnewcommand\Execute{\algexecutedesc{} }

\algnewcommand{\IIf}[1]{\State\algorithmicif\ #1\ \algorithmicthen}
\algnewcommand{\EElse}[1]{\algorithmicelse}
\algnewcommand{\EndIIf}{}

\algnotext{EndFor}
\algnotext{EndIf}
\algnotext{EndUpon}
\algnotext{EndFunction} 
\algnotext{EndProcedure} 
\algnotext{EndWhile}

\newcommand{\algrule}[1][.2pt]{\par\vskip.5\baselineskip\hrule height
#1\par\vskip.5\baselineskip}

%% file: introduction.tex
\section{Introduction}\label{sec:intro}

\emph{Replication} is an established method to increase service availability
and dependability. It means deployment of a service on multiple machines and
coordination of their actions so that a consistent state is maintained across
all the service replicas. In case of a (partial) system failure operational
replicas continue to provide the service. 

We consider two basic models of service replication: \emph{State Machine 
Replication (SMR)} and \emph{Deferred Update Replication (DUR)}. In the 
SMR approach \cite{Lam78}, each client request is first ordered among all 
service replicas and then processed by each replica independently. Given that 
the service is deterministic and all requests are executed in the same order 
(sequentially) by every replica, all the replicas are in a consistent state. 
The total order is achieved using fully distributed, fault-tolerant protocols 
for distributed agreement such as Total Order Broadcast (TOB) \cite{DSU04}. In 
DUR, which is an optimistic multi-primary-backup approach \cite{CBPS10}, no 
replica coordination is required prior or during request execution. Instead, 
each request is handled by only one replica using an \emph{atomic transaction}. 
A transaction can run in parallel with any other transactions. DUR uses an 
atomic commitment protocol (based, e.g., on TOB) to ensure consistency upon 
transaction commit. If a \emph{conflict} is detected, i.e., a transaction read 
data modified by a concurrent but already committed transaction, the 
transaction revokes all changes it performed so far and subsequently restarts.

In our previous work \cite{WKK12} \cite{WKK16}, we analytically and 
experimentally compared the SMR and DUR schemes (both based on TOB). Our 
results show that, surprisingly, there is no clear winner--each approach has 
its advantages and drawbacks, and various factors such as workload type, 
parallelism on multicore CPUs, and network congestion have significant impact 
on performance of the SMR and DUR schemes. Also the schemes differ in the 
offered semantics. Most notably, the differences lie in support for 
non-deterministic operations and irrevocable operations, i.e., operations, 
whose effects cannot be rolled back, such as local system calls. DUR provides 
the support for non-deterministic operations but forbids irrevocable operations 
as a transaction may abort due to a conflict. SMR requires deterministic 
operations as they are executed by every replica independently. There are also 
significant differences in provided correctness guarantees: SMR typically 
guarantees linearizability \cite{HW90} whereas DUR provides update-real-time 
opacity \cite{KKW16}, a flavour of opacity \cite{GK10} which allows aborted and 
read-only transactions to operate on stale but still consistent data.


This insight has led us to an idea of combining SMR and DUR into the 
\emph{Hybrid Transactional Replication (HTR)} scheme, which we introduce in 
this paper. This way, we aim to achieve increased performance and more flexible 
semantics. Some requests (transactions) are better performed in the state 
machine (SM) mode, especially if they access many objects, result in large 
updates, or cause many conflicts (e.g., resizing and rehashing a hashtable). On 
the other hand, other transactions that can be easily executed concurrently 
benefit from execution in the deferred update (DU) mode. The execution mode is 
selected dynamically per transaction execution basis by an oracle. The oracle, 
which is supplied by the programmer and can rely on machine learning techniques 
for better flexibility, constantly monitors the system to determine which mode 
is optimal for a particular run of a transaction. Among the data gathered by 
the oracle are the duration of transaction execution, the latency of TOB, the 
size of messages, network congestion, and the system load. Since read-only 
transactions do not modify the system's state, they are always executed in the 
optimistic DU mode and commit locally without any inter-process 
synchronization.

As formally proven in the paper, HTR offers strong consistency guarantees 
which are similar to DUR's. More precisely, HTR satisfies update-real-time 
opacity \cite{KKW16}, a flavour of opacity \cite{GK10} which allows aborted 
and read-only transactions to operate on stale but still consistent data. 
Compared to DUR, HTR offers richer transactional semantics with support for 
irrevocable operations. In HTR, transactions with irrevocable operations are 
simply executed in the SM mode which ensures abort-free execution. 


To evaluate our ideas, we extended with HTR our optimistic distributed 
transactional memory (DTM) system called Paxos STM \cite{WKK12} \cite{WKK16}.
Paxos STM replicates all \emph{transactional objects} (objects shared by 
transactions) and maintains strong consistency of object replicas. Transactions 
are executed atomically and in isolation despite system failures, such as 
server crashes; the crashed servers can be recovered. For expressiveness, 
transactions can be rolled back or retried on demand using the \emph{rollback} 
and \emph{retry} constructs. The latter one can be used in programming idioms 
such as suspending the execution until a given condition is met. To our best 
knowledge, Paxos STM is the first replicated DTM system to provide support for 
irrevocable operations within transactions.

We discuss techniques useful in the process of designing an oracle policy that 
matches the expected workload. To facilitate automatic adaptation to changing 
conditions, we propose the \emph{HybirdML} oracle (\emph{HybML} in short), a 
simple yet surprisingly robust oracle for Paxos STM which relies on online 
machine learning. HybML treats choosing the optimal mode for transaction 
execution as a multi-armed bandit problem for every class of transaction 
defined by the programmer. 


We compare the performance of HybML against two simple oracles that execute 
all updating transactions in either DU or SM mode. We examine scalability of 
the system under various workloads and show that HybML allows the system to 
achieve up to 50\% improvement in performance compared to the DU-only or 
SM-only oracles.
We also show that HybML quickly adapts to changing workloads. The results 
clearly indicate that in all cases an application can benefit from the HTR 
scheme. 

\subsection{Motivations and contributions}

The motivations to conduct this research were threefold. Firstly, as our 
previous work \cite{WKK12} \cite{WKK16} showed that neither SMR nor DUR scheme 
was superior, we were eager to combine these two into a single algorithm to 
bring together the best of both worlds. Secondly, we are not aware of any prior 
research on applying transactional semantics to state machine replication for 
increased expressiveness. Contrary to pure SMR, we achieve greater 
expressiveness by incorporating the \emph{rollback} and \emph{retry} 
constructs: they enable revoking changes performed by a request and restarting 
the execution of a transaction if required. Thirdly, to our best knowledge our 
research is the first on irrevocable actions in a replicated DTM.

The main contributions of the paper are as follows:
\begin{itemize}
\item We proposed a novel scheme called Hybrid Transactional Replication
(HTR), which combines state-machine--based and deferred-update replication 
schemes for better performance, scalability, and improved code expressiveness; 
the algorithm leverages transactional semantics and provides update-real-time 
opacity as a consistency criterion, as formally proven;

\item We developed HTR-enabled Paxos STM, a tool for hybrid transactional
replication of services;

\item We introduced an ML-based oracle for Paxos STM, which allows 
the system to automatically adapt to changing workloads;

\item We evaluated throughput and scalability of Paxos STM under various 
workloads when executing all updating requests either in the SM or DU modes, or 
the combination of these two when using the ML-based oracle. We also 
demonstrated how the ML-based oracle adapts to changing workloads.

\item We showed when a replicated service can benefit from HTR and discussed 
some techniques on how to configure the HTR algorithm for higher performance.
\end{itemize}

This paper is an extended version of our first paper on HTR \cite{KKW13}. In 
this publication, we include a formal proof of correctness for HTR and 
introduce and evaluate machine learning techniques for HTR's oracle.

\subsection{Paper structure}

The paper has the following structure. Firstly, we present related work in
Section~\ref{sec:related_work}. Next, we briefly discuss the SMR and DUR models 
in Section~\ref{sec:context}. Then, in Section~\ref{sec:htr}, we present the 
HTR algorithm and discuss its characteristics. Next, in 
Section~\ref{sec:htr_evaluation}, we show the results of the evaluation of 
HTR-enabled Paxos STM by comparing its performance and scalability under 
diverse workloads and oracles. Finally, we conclude with Section~\ref{sec:conc}.

%% file: related_work.tex
\section{Related Work} \label{sec:related_work}

%
%
%

%
%
%
%
%
%
%
%

In this section we present work relevant to our research. 

\subsection{Transactional replication}

Over the years, multiple data and service replication techniques have emerged 
(see \cite{CBPS10} for a survey). They differ in offered semantics, as well 
as performance under various workloads. We focus on replication schemes that 
offer strong consistency. 

\emph{State Machine Replication (SMR)} \cite{Lam78} \cite{Sch90} \cite{Sch82}
\cite{Sch93} (described in detail in Section~\ref{sec:context:smr}), is the
simplest and most commonly used non-transactional replication scheme. SMR uses
a distributed agreement protocol to execute client requests on all processes
(replicas) in the same order. For replica coordination, various fault-tolerant
synchronization algorithms for totally ordering events were proposed (see,
e.g., \cite{Lam78a} \cite{Lam84} among others). More recently, Total Order
Broadcast has been used for request dissemination among replicas (see
\cite{DSU04} for a survey of TOB algorithms and \cite{CBPS10} for further
references). Also implementations of TOB with optimistic delivery of messages
are used to build systems based on state replication, e.g., \cite{KPA+03}
\cite{PQR10} \cite{PQR11} \cite{HPR14a}.

\emph{Deferred Update Replication} \cite{CBPS10}, described in detail in 
Section~\ref{sec:context:dur}, is a basic replication scheme which features 
\emph{transactional semantics}. DUR is based on a multi-primary-backup approach 
which, unlike SMR, allows multiple (updating) client requests to be 
executed concurrently (as atomic transactions). Various flavors of DUR are 
implemented in several commercial database systems, including Ingres, MySQL 
Cluster and Oracle. These implementations use 2PC \cite{BHG87} as the atomic 
commitment protocol. In this work, we consider DUR based on TOB \cite{DSU04}. 
This approach is advocated by several authors because of its non-blocking 
nature and predictable behaviour (see \cite{PGS98} \cite{PGS03} \cite{AAeAS97} 
among others). Most recently, it has been implemented in D2STM \cite{CRCR09} 
and in our system called Paxos STM \cite{WKK12} \cite{KKW13} (characterised in 
Section~\ref{sec:htr_evaluation}). It has also been used as part of the 
coherence protocols of S-DUR \cite{SPJ12} and RAM-DUR \cite{DP12}.


In our previous work \cite{WKK12} \cite{WKK16}, we compared SMR and DUR both 
theoretically and practically and showed that neither scheme is superior in 
general. We discuss the differences between the schemes in 
Section~\ref{sec:context:comp}. 

HTR is also a strongly consistent replication scheme based on TOB. 
As we discuss in Section~\ref{sec:htr}, HTR shares many similarities with both 
SMR and DUR by allowing transactions to be executed either in a pessimistic or 
optimistic mode, resembling request or transaction execution in SMR or DUR, 
respectively. As SMR and DUR, HTR fits the framework of \emph{transactional 
replication (TR)} \cite{KKW15}, which formalizes the interaction between 
clients and the replicated system. The programming model of TR corresponds to 
\emph{Distributed Transactional Memory (DTM)}, as discussed below. HTR offers 
guarantees on transaction execution, which are similar to those provided by DUR 
(see Section~\ref{sec:htr:correctness}).

There are a number of optimistic replication protocols that, similarly to HTR, 
have their roots in DUR. For example, in \emph{Post\-gres-R} \cite{KA00}, TOB 
is only used to broadcast the updates produced by a transaction; a decision 
regarding transaction commit or abort is sent in a separate (not ordered) 
broadcast. \emph{PolyCert} \cite{CPR11} 
can switch between three TOB-based certification protocols, which differ in the 
way the readsets of updating transactions are handled. \emph{Executive DUR 
(E-DUR)} \cite{KKW14} streamlines transaction certification with the leader of 
the Paxos protocol. 
Note that all the above mentioned protocols are aimed only at increasing the 
throughput of DUR and not at extending the transactional semantics of the base 
protocol, as in case of HTR (see also 
Section~\ref{sec:related_work:txsemantics}). 

\subsection{Distributed Transactional Memory Systems} 
\label{sec:related_work:dtm}

The model of replication considered in this work closely corresponds to some
Distributed Transactional Memory systems.
DTM evolved as an extension of local (non-distributed) transactional memory 
\cite{HM93} to distributed environment. In TM, transactions are used to 
synchronize accesses to shared data items and are meant as an alternative to 
lock-based synchronization mechanisms. TM also has been proposed as an 
efficient hardware-supported mechanism for implementing monitors \cite{YVDH+14}.

We added the HTR functionality to Paxos STM \cite{WKK12}, which is an 
object-based DTM system that we developed to compare SMR and DUR and then  
used as a testbed for the E-DUR scheme \cite{KKW14}. It builds on JPaxos 
\cite{KSZ+11}--a highly optimized implementation of the Paxos 
algorithm \cite{Lam98}. 

Several other DTM systems were developed so far, e.g., Anaconda \cite{KLA+10}, 
Clu\-ster-STM \cite{BAC08}, DiSTM \cite{KAJ+08}, Hyflow \cite{SR11a} and 
Hyflow2 \cite{TRP13}. Notably, our system was designed from ground up as a 
fully distributed, fault-tolerant system, in which crashed replicas can 
recover. Unlike DiSTM, there is no central coordinator, which could 
become a bottleneck under high workload. The TOB-based transaction 
certification protocol implemented by Paxos STM simplifies the architecture, 
limits the number of communication steps and avoids deadlocks altogether 
(unlike the commit protocols in Anaconda or Hyflow/Hyflow2).
The use of TOB also helps with graceful handling of replica crashes (which, 
e.g., are not considered in Cluster-STM). The closest design to ours is the one 
represented by D2STM \cite{CRCR09}, which also employs full replication and 
transaction certification based on TOB. However, unlike Paxos STM, D2STM does 
not allow replicas to be recovered after crash nor transactions to contain 
irrevocable operations.

\subsection{Transaction semantics} \label{sec:related_work:txsemantics}




As mentioned earlier, HTR allows irrevocable operations in transactions 
executed in the SM mode, which are guaranteed to commit (see 
Section~\ref{sec:htr:specification}). The problem of irrevocable operations has 
been researched in the context of non-distributed TM (see e.g., \cite{BDL+07} 
\cite{OCS07} \cite{SMS08} \cite{WSAA08} among other). These operations are 
typically either forbidden, postponed until commit, or switched into an 
\textit{ad hoc} pessimistic mode \cite{UBES10}. Some solutions for starved 
transactions (i.e., transaction, which repeatedly abort) are relevant here, 
e.g., based on a global lock \cite{KR81} or leases \cite{CRR10}. The former is 
not optimal as it impacts the capability of the system to process transactions 
concurrently (unlike in HTR, where multiple transactions in the DU mode can 
execute concurrently with a transaction in the SM mode). On the other 
hand, the latter solution does not guarantee abort-free execution and requires 
a transaction to be first executed fully optimistically at least once. More 
recently, Atomic RMI, a fully-pessimistic DTM system, which provides support 
for irrevocable operations has been presented in \cite{SW15} and \cite{WS16}. 
Unlike our system, in which transactions are local in scope and data 
is consistently replicated, Atomic RMI implements distributed transactions and 
does not replicate data across different machines.

In database systems, there exists work on allowing nondeterministic operations, 
so also irrevocable operations. In \cite{TA10}, a centralized preprocessor is 
used to split a transaction into a sequence of subtransactions that are 
guaranteed to commit. Each subtransaction requires a separate broadcast, thus 
significantly increasing latency in transaction execution.


\subsection{Protocol switching} \label{sec:related_work:htr:protocol_switching}

Since in HTR a transaction can be executed in two different modes, solutions 
which allow for protocol switching are relevant. For example, PolyCert 
\cite{CPR11} features three certification protocols that differ in the way the 
readsets of updating transactions are handled. 
Morph-R \cite{CRRR13} features three interchangeable replication 
protocols (primary-backup, distributed locking based on 2PC, and 
TOB-based certification), which can be switched according to the current needs.
Contrary to PolyCert and Morph-R, our approach aims at the ability to execute 
transactions in different modes with the mode chosen on per-transaction-run 
basis. Additionally, our system considers a much wider set of parameters and 
can be tuned by the programmer for the application-specific characteristics. 
Hyflow \cite{SR11} allows various modes of accessing objects needed by a 
transaction: migrating them locally and caching (data flow) or invoking remote 
calls on them (control flow).
StarTM \cite{DDSL12} uses static code analysis to select between the execution 
satisfying snapshot isolation (SI) and serializability for increased 
performance.

In AKARA \cite{CPO08}, a transaction may be executed either by all 
replicas as in SMR, or by one replica with updates propagated after transaction 
finishes execution, in a somewhat similar way to which it is done in DUR. In 
the latter case, execution can proceed either in an optimistic or in a 
pessimistic fashion, according to a schedule established prior to transaction 
execution using conflict classes. However, in both cases the protocol requires 
two broadcast messages for every transaction: a TOB message to establish the 
final delivery order and a reliable broadcast message with the transaction's 
updates (DUR and HTR require only one broadcast for every transaction). Unlike 
in HTR, in AKARA the execution mode is predetermined for every transaction and 
depends on the transaction type.


Approaches that combine locks and transactions are also relevant. In 
\cite{WHJ06}, Java monitors can dynamically switch between the lock-based and 
TM-based implementations. Similarly, adaptive locks \cite{UBES10} enable 
critical sections that are protected either by mutexes or executed as 
transactions. However, the above two approaches use a fixed policy. In our 
approach, the HTR oracles implement a switching policy that can adapt to 
changing conditions.


\subsection{Machine learning techniques} \label{sec:related_work:ml}

The mechanisms implemented in our ML-based oracle for HTR are heavily inspired 
by some algorithms well known in the ML community. Most importantly, HybML
implements a policy that is similar to the \emph{epsilon-greedy strategy} for 
the \emph{multi-armed bandit problem} (see \cite{R52} for the original 
definition of the problem, \cite{LR85} for the proof of convergence, and 
\cite{KP14} for the survey of the algorithms solving the problem). However, 
some crucial distinctions can be made between the original approach and ours. 
We discuss them in detail in Section~\ref{sec:htr:ml}.

A survey of self-tuning schemes for the algorithms and parameters used in 
various DTM systems can be found in \cite{CDRR15}. A few ML-based mechanisms 
have been used in some of the transactional systems we discussed before. 
PolyCert \cite{CPR11} implements two ML approaches to select the optimal 
certification protocol. The first is an offline approach based on regressor 
decision trees, whereas the second uses the \emph{Upper Confidence Bounds} 
algorithm, typically used in the context of the multi-armed bandit problem. 
Because the used certification protocols behave differently under various 
workloads, in the latter approach the authors decided to discretize the 
workload state space using the size of readsets generated during execution of 
transactions. This differs from our approach, since in HybML we solve the 
multi-armed bandit problem independently for every class of transactions. The 
rough classification, which can be much finer than in PolyCert, is provided by 
the user. Morph-R \cite{CRRR13} uses three different black-box offline learning 
techniques to build a prediction model used to determine the optimal 
replication schemes for the current workload, i.e., decision-trees, neural 
networks, and support vector machines. Such heavy-duty ML approaches are not 
suitable for our purposes because HTR selects an execution mode for each 
transaction run independently and not for the whole system once every several 
minutes, as it is usually the case in typical applications of ML techniques 
(see, e.g., \cite{RSCQ12} \cite{CGR+11}). Hyflow \cite{SR11a} uses heuristics 
to switch between the data-flow and control-flow modes, but the authors do not 
provide details on the mechanisms used.

%% file: context.tex
\section{The Context of HTR} \label{sec:context}

In this section, we describe the context for the HTR algorithm.
We begin with the description of the system model. Then, we present the SMR and 
DUR replication schemes (we follow the description of algorithms from 
\cite{KKW15} and \cite{KKW16}). Finally, we briefly discuss 
strengths and weaknesses of both approaches.


\subsection{System model}

The model consists of a set $\mathcal{P} = \{ p_1, p_2, ..., p_n \}$ of $n$ 
\emph{service processes} (\emph{replicas}) running on independent machines 
(nodes) connected via a network. The processes communicate only by means of 
messages. External entities (\emph{clients}) issue \emph{requests} (also called 
\emph{transactions}) to any of the replicas and receive responses once the 
requests are processed. A client can issue only one request at a time. A 
request consists of a unique identifier $\id$, a program $\prog$ to be executed 
and arguments $\args$, which are necessary to execute the program. Some 
requests may be marked as \emph{read-only (RO)}, i.e., they do not alter the 
system's state. \emph{Updating} requests (also called \emph{read-write (RW) 
requests}) may or may not contain operations that modify the system's state. We 
assume a \emph{crash-recovery failure model}, where crash of at most $\lceil 
\frac{n}{2} \rceil - 1$ processes is tolerated. After \emph{recovery}, a failed 
process can rejoin the system at any time. We also assume availability of a 
failure detector $\Omega$ which is the weakest failure detector capable of 
solving distributed consensus in the presence of failures \cite{CHT96}. The 
discussed algorithms are memory model agnostic, i.e., they can be used in 
either the object-- or memory-word--based environments. To match our 
implementation, we assume an object-oriented memory model.

\subsection{State Machine Replication} \label{sec:context:smr}

In SMR \cite{Lam78} \cite{Sch90} \cite{Sch82}, for which we give the 
pseudocode in Algorithm~\ref{alg:smr}, a service is fully replicated by 
every process. Each client request, which can be handled by any replica, 
consists of three elements: a unique $\id$, $\code$, which 
specifies the operations to be executed and $\args$, which holds the arguments 
needed for the program execution. Prior to execution, the request is broadcast 
to all replicas using $\tobcast$ (line \ref{alg:smr:tob}). Only then each 
replica executes the request independently (line \ref{alg:smr:execute}). After 
the request is executed, the thread that originally received the request 
returns the response to the client (line \ref{alg:smr:return}). Note that for 
brevity we omit in the pseudocode some data structures holding a history of 
clients' requests, which have to be maintained to provide fault-tolerance in 
case of the loss of request/response messages.

Since all replicas start from the same initial state and process all requests 
in the same order (thanks to the properties of TOB), the state of the service 
is replicated on all machines. Naturally, execution of each request 
has to be deterministic. Otherwise, the consistency among replicas could not be 
preserved as the replicas might advance differently.

Note that SMR does not differentiate between read-only and updating requests.
In an optimized version, SMR may execute multiple read-only requests 
concurrently, with no inter-replica synchronization \cite{R11}. In practice,
however, such implementations are not common due to their increased complexity
and limited scalability (all updating requests still need to be executed 
sequentially by all replicas).

{\renewcommand\small{\scriptsize}%
\input{smr_alg.tex}
}

\subsection{Deferred Update Replication} \label{sec:context:dur}


Now we focus on DUR \cite{CBPS10} in its basic version, in which all data items 
(shared objects) managed by the replicated service are fully replicated on each 
replica. We give the pseudocode for DUR in Algorithm~\ref{alg:dur}.

Unlike in SMR, in DUR, every request (transaction) is executed only by 
a single replica and all replicas can execute different requests concurrently 
(also in separate threads on multiple processor cores). From the client's point 
of view, each transaction runs sequentially with respect to any other 
concurrent transactions in the system. Transaction execution happens 
optimistically and in isolation on local copies of shared objects 
(line~\ref{alg:dur:execute}). Additionally, all accesses to shared objects are 
recorded independently for each transaction (in the $\readset$, 
line~\ref{alg:dur:read}, and the $\updates$ set, line~\ref{alg:dur:write}, both 
kept as a part of the transaction descriptor, line~\ref{alg:dur:txDesc}). This 
information is disseminated among replicas (using TOB, line \ref{alg:dur:tob}), 
once 
the transaction enters the committing phase (calls the $\textsc{commit}$ 
procedure, line~\ref{alg:dur:commit}). It is the only moment in a transaction's 
lifetime that requires replica synchronization. 

Upon delivery of a message with state updates and information about accesses to 
shared objects performed by a transaction (line~\ref{alg:dur:adeliverTx}), each 
replica independently \emph{certifies} the transaction 
(line~\ref{alg:dur:globalCert}). It means that the replica checks whether the 
committing transaction had read any shared objects modified by a concurrent but 
already committed transaction (see below how DUR establishes the precedence 
order between transactions). This is done by comparing the $\updates$ sets of 
the already committed updating transactions (stored in the $\Log$ variable), 
with the $\readset$ of the committing transaction. If none of the sets 
intersect, the transaction commits, the state modification it produced are made 
visible (lines \ref{alg:dur:commitBegin}--\ref{alg:dur:commitEnd}), and the 
response is returned to the client (line~\ref{alg:dur:endreq}). Otherwise the 
transaction is rolled back and restarted (line~\ref{alg:dur:restart}). 

DUR establishes the precedence order between transaction, using a logical 
clock. To this end, each process of DUR maintains a global variable $\LC$, 
which is incremented every time a process applies updates of a transaction 
(line \ref{alg:dur:inc}). $\LC$ is used to mark the start and the end of the 
transaction execution (lines \ref{alg:dur:start} and \ref{alg:dur:end}). For 
transaction descriptors $t_i$ and $t_j$ of any two transactions $T_i$ and 
$T_j$ ($T_i \neq T_j$ ) in any execution of DUR, if $t_i.\eend \leq 
t_j.\start$, then $T_i$ preceeds (in real-time) $T_j$;
otherwise $T_i$ and $T_j$ are concurrent. $\LC$ also allows the process to 
track whether its state is recent enough to execute the client's request (line 
\ref{alg:dur:clock}). For that we require that there is an additional element 
passed along with every client request (the $\clock$ variable, 
line~\ref{alg:dur:clock}) and every client response (the current value of the 
$\LC$ variable, line~\ref{alg:dur:endreq}, see below).

To ensure that a live transaction always executes on a consistent state, we 
perform partial transaction certifications upon every read operation 
(line~\ref{alg:dur:check}). However, these procedures are done only locally and 
do not involve any inter-replica synchronization. 

Read-only transactions, i.e., transactions that did not perform any 
updating operations (line~\ref{alg:dur:noupdates}), do not require 
inter-process synchronization in order to commit, because their execution does 
not result in any changes to the local or replicated state. Once a read-only 
transaction finishes execution, it can safely commit straight away 
(line~\ref{alg:dur:noupdates2}). All possible conflicts would have been 
detected earlier, upon read operations (line \ref{alg:dur:check}). Note that 
for updating transactions, we perform an additional certification just prior to 
broadcasting its transaction descriptor (line~\ref{alg:dur:localCert}). This 
step is not mandatory, but allows the process to detect conflicts earlier, and 
thus sometimes avoids costly network communication.

To manage the control flow of a transaction, the programmer can use two 
additional procedures: \textsc{rollback} (line \ref{alg:dur:rollback}), which 
stops the execution of a transaction and revokes all the changes it performed 
so far, and \textsc{retry} (line \ref{alg:dur:retry}), which forces a 
transaction to rollback and restart. 

For clarity, we make several simplifications. Firstly, we use a single global 
(reentrant) lock to synchronize operations on $\LC$ (lines \ref{alg:dur:start}, 
\ref{alg:dur:inc}, \ref{alg:dur:end}), $\Log$ (lines \ref{alg:dur:logFilter} 
and \ref{alg:dur:logAppend}) and the accesses to transactional objects (lines 
\ref{alg:dur:getObjectCall} and \ref{alg:dur:apply}). Secondly, we allow 
$\Log$ to grow indefinitely. $\Log$ can easily be kept small by garbage 
collecting information about the already committed transactions that ended 
before the oldest live transaction started its execution in the system. 
Thirdly, we use the same certification procedure for both the certification 
test performed upon every read operation (line \ref{alg:dur:check}) and the 
certification test that happens after a transaction descriptor is delivered to 
the main thread (line \ref{alg:dur:globalCert}). In practice, doing so would be 
very inefficient, because upon every read operation we check for the 
conflicts against all the concurrent transactions (line 
\ref{alg:dur:logFilter}), thus performing much of the same work again and 
again. However, these repeated actions can be easily avoided by associating the 
accessed shared objects with a version number equal to the value of $\LC$ at 
the time the objects were most recently modified.

{\renewcommand\small{\scriptsize}%
\input{dur_alg_shorter2.tex}
}

\subsection{SMR vs DUR comparison} \label{sec:context:comp}





In many cases SMR proves to be highly efficient although it allows no 
parallelism (or limited parallelism in its optimized version). In fact, when a 
workload is not CPU intensive, it performs much better than DUR \cite{WKK12} 
\cite{WKK16}. Also SMR is relatively easy to implement, because most of the 
complexity is hidden behind TOB. A major drawback of SMR is that it requires a 
replicated service to be deterministic. Otherwise consistency could not be 
preserved.
%


Contrary to SMR, in DUR parallelism is supported for 
read-only as well as updating transactions by default--each transaction is 
executed by a single replica in a separate thread and in isolation. This way 
DUR takes better advantage over modern multicore hardware. However, the 
performance of DUR is limited for workloads generating high contention. It is 
because in such conditions transactions may be aborted numerous times before 
eventually committing. Aborting live transactions as soon as they are known to 
be in conflict with a transaction that had just recently committed may help but 
only to some degree. 

DUR requires no synchronization (no communication step) among replicas 
for read-only transactions as they do not change the local or replicated state. 
This way read-only requests are handled by DUR much more efficiently 
compared to SMR. Additionally, in DUR read-only transactions can be provided 
with abort-free execution guarantee by introducing the multiversioning scheme. 
\cite{BG83} \cite{KKW15}. Multiversioning allows multiple versions of all 
transactional objects to be stored while being transparent to the programmer, 
i.e., at any moment only one version of any transactional object is accessible 
by a transaction. Paxos STM, which we extended with the HTR algorithm presented 
in this paper, implements both early conflict detection as well as the 
multiversioning scheme.

Usually there is a significant difference in the size of network messages 
communicated between replicas in SMR and DUR. In DUR, the broadcast 
messages contain transaction descriptors with readsets and updates sets. The 
size of these messages can be significant even for a medium sized transaction. 
Large messages cause strain on the TOB mechanism and increase transaction 
certification overhead. On the other hand, in SMR usually the requests consist 
only of an identifier of a method to be executed and data required for its 
execution; these messages are often as small as 100B.



DUR supports concurrency on multicore architectures. Concurrent programming
is error-prone but atomic transactions greatly help to write correct programs.
Firstly, operations defined within a transaction appear as a single logical 
operation whose results are seen entirely or not at all. Secondly, concurrent 
execution of transactions is deadlock-free which guarantees progress.
Moreover, the 
\emph{rollback} and \emph{retry} constructs enhance expressiveness. However, as 
mentioned earlier, irrevocable operations are not permitted since at any moment 
a transaction may be forced to abort and restart due to conflicts with other 
transactions.

Both SMR and DUR offer strong consistency guarantees. SMR ensures 
linearizability \cite{HW90}. On the other hand, DUR guarantees update-real-time 
opacity \cite{KKW15a} \cite{KKW16}, a flavour of opacity \cite{GK10} which 
allows aborted and read-only transactions to operate on stale but still 
consistent data. As shown in \cite{KKW16}, when transactions are hidden from 
clients, DUR provides update-real-time linearizability which is strictly weaker 
than (real-time) linearizability offered by SMR (requests that modify the 
system's state are provided the same guarantees in both SMR and DUR).

%% file: smr_alg.tex
\begin{algorithm}[b]\footnotesize
\caption{State Machine Replication for process $p_i$}
\label{alg:smr}      
\footnotesize
\begin{algorithmic}[1]
\item[\textbf{Thread $q$ on request} $r$ \textbf{from client} $c$ (executed on 
one replica)]
\State{response $\res_q \gets \bot$}
\Upon{init}{}
  \State \tobcast $r$ \label{alg:smr:tob}\hfill\textit{// blocking}
  \State \EndReqBasic \label{alg:smr:return}
\EndUpon
\algrule
\item[\textbf{The main thread of SMR} (executed on all replicas)]
\Upon{\todeliver}{request $r$} \label{alg:smr:adeliver}
  \State{response $\res \gets$ execute $r.\prog$ with $r.\args$} \label{alg:smr:execute}
  \If{request with $r.\id$ handled locally by thread $q$}
     \State{$\res_q \gets res$}
  \EndIf
\EndUpon
\end{algorithmic}
\end{algorithm}

%% file: dur_alg_shorter2.tex
\begin{algorithm}[t] 
\caption{Deferred Update Replication for process $p_i$}
\label{alg:dur}  
  \footnotesize
\begin{algorithmic}[1]
\State{integer $\LC \gets 0$} \label{alg:dur:lc}
\State{set $\Log \gets \emptyset$} \label{alg:dur:log}

\Function{getObject}{txDescriptor $t$, objectId $oid$}
  \If{$(\oid, \obj) \in t.\updates$} $\vvalue \gets \obj$ \label{alg:dur:retrieveUpdates}
  \State{\textbf{else} {\ $\vvalue \gets$ retrieve object $\oid$} \label{alg:dur:retrieve}} \EndIf
  \State{\Return $\vvalue$}
\EndFunction

\Function{certify}{integer $\start$, set $\readset$} \label{alg:dur:certify}
  \State{\LockStart $L \gets \{ {t \in \Log} : t.\eend > \start \}$ \LockEnd} \label{alg:dur:logFilter}
  \ForAll{$t \in L$}
      \State{$\writeset \gets \{ \oid : \exists (\oid,\obj) \in t.\updates \}$}
      \IIf{$\readset \cap \writeset \neq \emptyset$} \Return $\failure$ \EndIIf
  \EndFor
  \State{\Return $\success$}
\EndFunction

\algrule

\item[\textbf{Thread $q$ on request} $r$ \textbf{from client} $c$ (executed on 
one replica)]
\State{enum $\outcome_q \gets \failure$}\hfill\textit{// type: enum \{ success, failure \}}
\State{response $\res_q \gets \nnull$}

\State{txDescriptor $t \gets \nnull$}\hfill\textit{// type: record (id, start, end, readset, 
updates)} \label{alg:dur:txDesc}

\Upon{init}{}
  \State{\textbf{wait until} $\LC \ge r.\clock$} \label{alg:dur:clock}
  \State{\Call{transaction()}{}}
  \State{\EndReq} \label{alg:dur:endreq}
\EndUpon

\Procedure{transaction}{}
  \State{$t \gets (\text{a new unique}\ \id, 0, 0, \emptyset, \emptyset)$}
  \State{\LockStart $t.\start \gets LC$ \LockEnd} \label{alg:dur:start}
  \State{$\res_q \gets$ execute $r.\prog$ with $r.\args$} \label{alg:dur:execute}
  \IIf{$\outcome_q = \failure$} \Call{transaction()}{} \label{alg:dur:restart} \EndIIf 
\EndProcedure

\Function{read}{objectId $\oid$}
\State{$t.\readset \gets t.\readset \cup \{\oid\}$} \label{alg:dur:read}
\State{\LockStart \textbf{if} \Call{certify}{$t.\start, \{\oid\}$} $= 
\failure$ \textbf{then}} \label{alg:dur:check} \Call{retry()}{}
     \State{\AlgIndent\AlgIndentSmall \textbf{else} \Return \Call{getObject}{$t$, $\oid$}} \label{alg:dur:getObjectCall}
\LockEnd
\EndFunction

\Procedure{write}{objectId $\oid$, object $\obj$}
\State{$t.\updates \gets \{ (\oid',\obj') \in t.\updates : \oid' \neq \oid \} \cup \{(\oid,\ \obj)\}$} \label{alg:dur:write}
\EndProcedure

\Procedure{commit}{} \label{alg:dur:commit}
  \State \EndTxCode
  \If{$t.\updates = \emptyset$} \label{alg:dur:noupdates}
    \State{$\outcome_q = \success$}
    \State \Return \label{alg:dur:noupdates2}
  \EndIf
  \IIf{\Call{certify}{$t.\start, t.\readset$} $= \failure$} \label{alg:dur:localCert} {\Return} \EndIIf
  \State{\tobcast $t$}\hfill\textit{// blocking} \label{alg:dur:tob}
\EndProcedure

%

\Procedure{retry}{} \label{alg:dur:retry}
  \State{\EndTxCode}
\EndProcedure

\Procedure{rollback}{} \label{alg:dur:rollback}
  \State{\EndTxCode}
  \State{$\outcome_q \gets \success$}
\EndProcedure

\algrule
\item[\textbf{The main thread of DUR} (executed on all replicas)]
\Upon{\todeliver}{txDescriptor $t$}  \label{alg:dur:adeliverTx}
  \If{\Call{certify}{$t.\start, t.\readset$} $ = \success$} \label{alg:dur:globalCert}
     \State{\LockStart $\LC \gets \LC + 1$ } \label{alg:dur:inc} \label{alg:dur:commitBegin}
     \State{\AlgIndent $t.\eend \gets \LC$ } \label{alg:dur:end} 
     \State{\AlgIndent $\Log \gets \Log \cup \{t\}$} \label{alg:dur:logAppend}
     \State{\AlgIndent apply $t.\updates$ \LockEnd} \label{alg:dur:apply} \label{alg:dur:commitEnd}
  \IIf{transaction with $t.id$ executed locally by thread $q$} $\outcome_q \gets \success$\EndIIf
  \EndIf 
\EndUpon
\end{algorithmic}
\end{algorithm}

%% file: htr.tex
\section{Hybrid Transactional Replication} \label{sec:htr}

In this section, we define Hybrid Transactional Replication (HTR), a novel 
transactional replication scheme that seamlessly merges DUR and SMR. First, we 
discuss the transaction oracle--the key new component of our algorithm. Next, 
we explain the HTR algorithm by presenting its pseudocode and giving the proof 
of correctness. Then, we briefly discuss the strengths of HTR. Finally, we 
present two approaches to creating an oracle: a manual, tailored for a given 
workload, and an automatic, based on machine learning.

\subsection{Transaction oracle}

Our aim was to seamlessly merge the SMR and DUR schemes, so that requests 
(transactions) can be executed in either scheme depending on the desired 
performance considerations and execution guarantees (e.g., support for 
irrevocable operations). \emph{Transaction oracle} (or \emph{oracle}, in short) 
is a mechanism that for a given transaction's run is able to assess the best 
execution mode: either the \emph{SM mode}, which resembles request execution 
using SMR, or the \emph{DU mode}, which is analogous to executing a request 
using DUR. The oracle may rely on hints declared by the programmer as well as 
on dynamically collected \emph{statistics}, i.e., data regarding various 
aspects of system's performance, such as:
\begin{itemize}
\item duration of various phases of transaction processing, e.g., 
      execution time of a request's (transaction's) code, TOB latency, 
      and duration of transaction certification,
\item abort rate, i.e., the ratio of aborted transaction runs to all 
      execution attempts, 
\item sizes of exchanged messages, readsets and updates sets,
\item system load, i.e., a measure of utilization of system resources such as
CPU and memory,
\item delays introduced by garbage collector,
\item saturation of the network.
\end{itemize}
Declared read-only transactions, i.e., transactions known \emph{a priori} to be 
read-only, are always executed in the DU mode since they do not 
alter the local or replicated state and thus do not require distributed 
certification. Hence, decisions made by the oracle only regard updating 
transactions.

Since the hardware and the workload can vary between the replicas the system 
can use different oracles at different nodes and independently change them at
runtime when desired. For brevity, in the description of the algorithm we
abstract away the details of the oracle implementation and treat it as a black
box with only two functions: $\textsc{feed}(\mathit{data})$, used to update the 
oracle with data collected over the last transaction's run, regardless of the 
outcome, and $\textsc{query}(\mathit{request})$, used to decide in 
which mode a new transaction is to be executed). 

The problem of creating a well-performing oracle is non-trivial and depends on 
the expected type of workload. In Section~\ref{sec:htr:tuning} we discuss a 
handful of tips on how to build an oracle that matches the expected workload. 
Then, in Section~\ref{sec:htr:ml}, we also show an oracle which uses machine 
learning techniques to automatically adjust its policy to changes in the 
workload.

\subsection{Specification} \label{sec:htr:specification}

Below we describe the HTR algorithm, whose pseudocode is given in 
Algorithm~\ref{alg:htr}. HTR is essentially DUR (Algorithm~\ref{alg:dur}), 
extended with the SMR scheme (Algorithm~\ref{alg:smr}) and the 
\textsc{updateOracleStatistics} procedure (line \ref{alg:htr:updateOracle}) 
that feeds the oracle with the statistics collected in a particular run of a 
transaction before the transaction is committed, rolled back, or retried.

Note that HTR features two sets of functions/procedures facilitating 
execution of a transaction (i.e., \emph{read} and \emph{write} operations on 
shared objects) and managing the control flow of the transaction (i.e., 
procedures used to commit, rollback or retry the transaction). One set of 
functions/procedures is used by transactions executed in the DU mode (lines 
\ref{alg:htr:DUFunctionsBegin}--\ref{alg:htr:DUFunctionsEnd}) and 
one is used by transactions executed in the SM mode (lines 
\ref{alg:htr:SMFunctionsBegin}--\ref{alg:htr:SMFunctionsEnd}).

When transaction is about to be executed, the oracle is queried to determine
the execution mode for this particular transaction run 
(line~\ref{alg:htr:query}). When the DU mode is chosen (line 
\ref{alg:htr:dumode}), a transaction, called a \emph{DU transaction}, is 
executed and certified exactly as in DUR. It means that it is executed locally
(line~\ref{alg:htr:executeDU}) and only once commit is attempted 
(line~\ref{alg:htr:commit}) and the transaction passes local certification 
(line~\ref{alg:htr:localCert}), it is broadcast using TOB 
(line~\ref{alg:htr:tob}) to all replicas to undergo the final certification. 
On the other hand, when the SM mode 
is chosen (line \ref{alg:htr:smmode}), the request is first broadcast using TOB 
(line \ref{alg:htr:tobSM}) and then executed on all replicas as an \emph{SM 
transaction} (lines 
\ref{alg:htr:executionStartSM}--\ref{alg:htr:executionEndSM}). The execution of 
SM transactions happens in the same thread which is responsible for certifying 
DU transactions and applying the updates they produced. It means that at most 
one SM transaction can execute at a time and its execution does not interleave 
with handling of commit of DU transactions. However, the algorithm does not 
prevent concurrent execution of an SM transaction and multiple DU transactions; 
only the certification test and the state update operations of these DU 
transactions may be delayed until the SM transaction is completed. 

Since the execution of an SM transaction is never interrupted by receipt of a 
transaction descriptor of a DU transaction, no SM transaction is ever aborted.
It means that an SM transaction does not need to be certified and can commit 
straight away (lines \ref{alg:htr:commitBeginSM}--\ref{alg:htr:commitEndSM}). 
For the same reason, reading a shared object does not involve checking for 
conflicts (line \ref{alg:htr:readSM}) and reading the current value of $\LC$ 
(line \ref{alg:htr:startSM}) does not have to be guarded by a lock.

Naturally, an SM transaction has to be deterministic, so that the state of the 
system is kept consistent across replicas.

Note that after every transaction's run (regardless of the used execution mode 
and the fate of the transaction, i.e., whether the transaction commits, aborts 
or is rolled back), the statistics gathered during the run are fed to the 
oracle (lines \ref{alg:htr:updateOracleDU} and \ref{alg:htr:updateOracleSM}).

Because the pseudocode of HTR is based on the pseudocode we provided for DUR 
(Algorithm~\ref{alg:dur}), there are similar simplifications in both 
pseudocodes: we use a single global (reentrant) lock to synchronize operations 
on $\LC$ (lines \ref{alg:htr:start}, \ref{alg:htr:inc}, \ref{alg:htr:end}, 
\ref{alg:htr:incSM}, \ref{alg:htr:endSM}), $\Log$ (lines 
\ref{alg:htr:logFilter}, \ref{alg:htr:logAppend}, \ref{alg:htr:logAppendSM}), 
and the accesses to transactional objects (lines \ref{alg:htr:getObjectCall}, 
\ref{alg:htr:apply}, \ref{alg:htr:applySM}), we allow $\Log$ to grow 
indefinitely and we use the same certification procedure for both the 
certification test performed upon every read operation for DU transactions 
(line \ref{alg:htr:check}) and the certification test that happens after a 
transaction descriptor of a DU transaction is delivered to the main thread 
(line~\ref{alg:htr:globalCert}). The limitations introduced by these 
simplifications can be mitigated in a similar manner as in DUR.

\newcommand{\tdu}{\mathit{t_{DU}}}
\newcommand{\tsm}{\mathit{t_{SM}}}
\newcommand{\odu}{\mathit{outcome_q}}
\newcommand{\osm}{\mathit{outcome_q}}

{\renewcommand\small{\scriptsize}%
\input{htr_alg_shorter.tex}
}

\subsection{Characteristics} \label{sec:htr:characteristics}

Below we present the advantages of the HTR algorithm compared to 
the exclusive use of the schemes discussed in Section~\ref{sec:context}. 
We also discuss the potential performance benefits that will be evaluated 
experimentally in Section~\ref{sec:htr_evaluation}.

\subsubsection{Expressiveness}

Implementing services using the original SMR replication scheme is 
straightforward since it does not involve any changes to the service code. 
However, the programmer does not have any constructs to express control-flow 
other than the execution of a request in its entirety. In our HTR replication 
scheme, the programmer can use expressive transactional primitives 
\textproc{rollback} and \textproc{retry} to withdraw any changes made by 
transactions and to retry transactions (possibly in a different replication 
mode). In this sense, these constructs are analogous to DUR's, but they are 
also applicable for transactions executed in the pessimistic SM mode.
Upon retry, the SM transaction is not immediately reexecuted on each 
node. Instead, the control-flow returns to the thread which is responsible for 
handling the original request. The oracle is then queried again, to determine 
in which mode the transaction should be reexecuted. Similarly, reexecution of 
DU transactions is also controlled by the oracle.

Constructs such as \textproc{retry} can be used to suspend execution of a 
request until a certain condition is met. Note that in SMR doing so is not 
advisable since it would effectively block the whole system. It is because in 
SMR all requests are executed serially in the order they are 
received. On the contrary, when \textproc{retry} is called from within an SM 
transaction, the HTR algorithm rolls back the transaction and allows it to be 
restarted when the condition is met.

\subsubsection{Irrevocable operations}

In DUR, transactions may be aborted and afterwards restarted due to conflicts 
with other older transactions. Thus, they are forbidden to perform 
\emph{irrevocable operations} whose side effects cannot be rolled back (such as 
local system calls). \emph{Irrevocable (or inevitable) transactions} are 
transactions that contain irrevocable operations. Support for such transactions 
is problematic and 
has been subject of extensive research in the context of non-distributed TM 
(see Section~\ref{sec:related_work:txsemantics}). However, the proposed methods 
and algorithms are not directly transferable to distributed TM systems where 
problems caused by distribution, partial failures, and communication must also 
be considered. Below we explain how the HTR algorithm deals with irrevocable 
transactions.


In the HTR algorithm, irrevocable transactions are executed exclusively in the 
SM mode, thus guaranteeing abort-free execution, which is necessary for 
correctness. It also means that only one irrevocable transaction is executed at 
a time. 
However, our scheme does not prevent DU transactions to be executed in 
parallel--only certification and the subsequent process of applying 
updates of DU transactions (in case of successful certification) must be 
serialized with execution of SM transactions. Since an SM transaction runs on 
every replica, we only consider deterministic irrevocable transactions. 
Non-deterministic transactions would require acquisition of a global 
lock or a token to be executed exclusively on a single replica. 
Alternatively, some partially centralized approaches could be employed, as in
\cite{TA10}. However, they introduce additional communication steps, increase
latency, and may force concurrent transactions to wait a significant amount 
of time to commit. 

We forbid the \textproc{rollback} and \textproc{retry} primitives
in irrevocable transactions (as in \cite{WSAA08} and other TM systems) 
since they may leave the system in an inconsistent 
state.\footnote{Interestingly, Atomic RMI \cite{WS16}, a fully pessimistic 
distributed (but not replicated) TM system, allows nondeterministic irrevocable 
operations to be performed inside transactions.}

\subsubsection{Performance}

As mentioned in Section~\ref{sec:context:smr}, it is not straightforward to 
optimize the original SMR scheme to handle read-only requests in parallel 
with other (read-only or updating) requests.
However, in the HTR algorithm, read-only transactions are executed only by one 
replica, in parallel with any updating transactions--there is no need for 
synchronization among replicas to handle the read-only transactions. 

HTR can benefit from the multiversioning optimization in the same way as 
can the DUR scheme. In HTR extended with this 
optimization read-only transactions are guaranteed abort-free execution thus
boosting HTR's performance for workloads dominated by read-only requests. The 
implementation of HTR, which we use in our tests, implements the 
multiversioning optimization (see Section~\ref{sec:htr_evaluation}).

Unless an updating transaction is irrevocable (thus executed in the SM mode) or
non-deterministic (thus executed in the DU mode), it can be handled by 
HTR in either mode for increased performance. The choice is made by the HTR 
oracle that constantly gathers statistics during system execution and can 
dynamically adapt to the changing workload (which may vary between the 
replicas). In Section~\ref{sec:htr:tuning}, we discuss the tuning of the 
oracle and in Section~\ref{sec:htr:ml} we introduce an oracle, which relies on 
machine learning techniques for dynamic adaptation to changing conditions.

\subsection{Correctness} \label{sec:htr:correctness}

Below we give formal results on the correctness of HTR. The reference safety 
property we aim for is update-real-time opacity which we introduced in 
\cite{KKW16} and used to prove correctness of DUR. Roughly speaking, 
update-real-time opacity 
is satisfied, if for every execution of an algorithm (represented by some 
history $H$) it is possible to construct a sequential 
history $S$ such that: 
\begin{enumerate}
\item $H$ is equivalent to $S$, i.e., $H$ and $S$ contain the same set of 
transactions, all read and write operations return the same values and the 
matching transactions commit with the same outcome, 
\item every transaction in $S$ is legal, i.e. the values of shared objects read 
by the transaction are not produced out of thin air but match the specification 
of the shared objects, and 
\item $S$ respects real-time order for committed updating transactions in $H$, 
i.e., for any two committed updating transactions $T_i$ and $T_j$, if $T_i$ 
ended before $T_j$ started then $T_i$ appears before $T_j$ in $S$.
\end{enumerate}

However, it is impossible to directly prove that HTR satisfies update-real-time 
opacity due to a slight model mismatch, as we now explain. Recall that in HTR, 
every time a request is executed in the SM mode, multiple identical 
transactions are executed across whole system (the transactions operate on the 
same state and produce the same updates). In the formalization of 
update-real-time opacity (which is identical to the formalization of the 
original definition of opacity by Guerraoui and Kapalka), every such 
transaction is treated independently. Therefore, unless such an SM transaction 
did not perform any modifications or rollback on demand, it is impossible to 
construct such a sequential history $S$, in which every transaction is 
legal.\footnote{As a counter example consider an execution of HTR featuring a 
single client request which is executed as an SM transaction on every replica: 
a transaction first reads $0$ from a transactional object $x$ and subsequently 
increments the value of $x$, i.e., writes $1$ to $x$. Even for 2 replicas, it 
is impossible to construct a legal sequential history featuring all the SM 
transactions.} However, we can show that execution of multiple SM transactions 
regarding the same client request is equivalent to an execution of a single 
transaction (on some replica) followed by dissemination of updates to all 
processes, as in case of a DU transaction. Therefore, we propose a mapping 
called \emph{SMreduce}, which allows us to reason about the correctness of HTR. 
Roughly speaking, under the SMreduce mapping of some history of HTR, for any 
group of SM transactions regarding the same request $r$, such that the 
processes that executed the transactions applied the updates produced by the 
transactions, we allow only the first transaction of the group in the 
history to commit; other transactions appear aborted in the transformed 
history. The detailed definition of SMreduce, together with formal proof of 
correctness can be found in Appendix~\ref{sec:app:htr}.

Before we prove that HTR satisfies update-real-time opacity under the SMreduce
mapping, we first show that HTR does not satisfy a slightly stronger property, 
\emph{write-real-time opacity}, and thus also does not guarantee 
\emph{real-time opacity} (which is equivalent to the original definition of 
opacity \cite{GK10}, as shown in \cite{KKW16}).

\begin{restatable}{theorem}{htrno}
\label{thm:htr_no}
Hybrid Transactional Replication does not satisfy write-real-time opacity.
\end{restatable}

\iftoggle{sketches}{
\begin{proof}
Trivially, every t-history of DUR is also a valid t-history of HTR, because 
transactions in DUR are handled exactly in the same way as DU transactions in 
HTR. Since DUR does not satisfy write-real-time opacity \cite{KKW16}, neither 
does HTR.
\end{proof}
}{}

\begin{restatable}{corollary}{htrnoro}
Hybrid Transactional Replication does not satisfy real-time opacity.
\end{restatable}

\iftoggle{sketches}{
\begin{proof}
The proof follows directly from Theorem~\ref{thm:htr_no} 
and 
definitions of write-real-time opacity and real-time opacity (real-time opacity 
is strictly stronger than write-real-time opacity).
\end{proof}
}{}

\begin{restatable}{theorem}{htro}
\label{thm:htr_o}
Under the SMreduce mapping, Hybrid Transactional Replication satisfies 
update-real-time opacity.
\end{restatable}

\newcommand{\comp}[1]{\expandafter\bar#1}

\iftoggle{sketches}{
\begin{proof}[Proof sketch]
In order to prove that HTR satisfies update-real-time opacity, we have to show 
that (by Corollary~1 of \cite{KKW16}) for every finite t-history $H$ 
produced by HTR, under SMreduce there exists a t-sequential t-history $S$ 
equivalent to $\comp{H}$ (some completion of $H$), such that $S$ respects the 
update-real-time order of $H$ and every transaction $T_k$ in $S$ is legal in 
$S$. 

The proof is somewhat similar to the proof of correctness of DUR. The crucial 
part of the proof concerns a proper construction of $S$. We do so 
in the following way. Let us start with an empty t-sequential t-history $S'$.
Now we add to $S'$ t-histories $\comp{H}|T_k$ of all committed updating 
transactions $T_k$ from $\comp{H}$. However, we do so according to the order of 
the (unique) numbers associated with each such transaction $T_k$ (with 
transaction descriptor $t_k$). The number associated with $T_k$ is the value 
of $t_k.\eend$. This value is equal to the value of the $\LC$ variable of the 
replica that processes the updates of $T_k$ upon committing $T_k$. We show that 
this number is the same for each replica and that it unambiguously identifies 
the operation execution (thanks to the use of TOB as the sole mechanism used to 
disseminate messages among replicas, the fact that $\LC$ is incremented upon 
commit of every updating transaction and the SMreduce mapping). Then, one 
transaction at a time, we add to $S'$ all operations of all aborted and 
read-only transactions from $\comp{H}$. Note that we include here all the SM 
transactions which are committed in the original history but aborted under the 
SMreduce mapping in $\comp{H}$. For every such transaction $T_k$ (with 
transaction descriptor $t_k$), we insert $\comp{H}|T_k$ in $S'$ immediately 
after operations of a committed updating transaction $T_l$ (with transaction 
descriptor $t_l$), such that $t_k.\start = t_l.\eend$. It means that $T_l$ is 
the last committed updating transaction processed by the replica prior to 
execution of $T_k$ (note that $T_l$ might not exist if $T_k$ is executed on the 
initial state of the replica). Now we consider independently each (continuous) 
sub-t-history of $S'$ which consists only of operations of read-only and 
aborted transactions (i.e., roughly speaking, we consider the periods of time 
in between commits of updating transactions). For each such sub-t-history, when 
necessary, we rearrange the operations of the read-only and aborted 
transactions so that their order respects the order in which these operation 
executions appear in sub-t-history $\comp{H}|p_i$, for every process $p_i$. 
Because we always consider all operations of every transaction together, $S'$ 
is t-sequential. Then $S = S'$.

By construction of $S$, it is easy to show that $S$ respects the 
update-real-time order of $H$. Then we show by a contradiction that there does 
not exist a transaction in $S$ that is not t-legal, hence every transaction in 
$S$ is legal. 
\end{proof}
}{}

\subsection{Tuning the oracle} \label{sec:htr:tuning}

As pointed out in \cite{RCR08}, DTM workloads that are commonly considered 
are usually highly diversified in regard to the execution times and to the 
number of objects accessed by each transaction (this is also reflected in our 
benchmark tests in Section~\ref{sec:htr_evaluation}). However, the execution 
times of the majority of transactions are way under 1~ms. Therefore, the 
mechanisms that add to transaction execution time have to be lightweight or 
otherwise the benefits of having two execution modes will be overshadowed by 
the costs of maintaining an oracle.

In the HTR algorithm, the oracle is defined by only two methods that have to be 
provided by the programmer. Combined with multiple parameters collected by the 
system at runtime, the oracle allows for a flexible solution that can be tuned 
for a particular application. Our experience with HTR-enabled Paxos STM and 
multiple benchmarks shows that there are the two most important factors 
that should be considered when implementing an oracle:

\begin{itemize}
\item Keeping abort rate low. A high abort rate means that many transactions
executed in the DU mode are rolled back (multiple times) before they
finally commit. This undesirable behaviour can be prevented by executing 
some (or all) of them in the SM mode. The SM mode can also be chosen for 
transactions consisting of operations that are known to generate a lot of 
conflicts, such as resizing a hashtable. On the contrary, the DU mode is 
good for transactions that do not cause high contention, so can be executed 
in parallel thus taking advantage of modern multicore hardware.

\item Choosing the SM mode for transactions that are known to generate large
messages when executed optimistically in the DU mode. Large messages increase 
network congestion and put strain on the TOB mechanism, thus decreasing its 
performance. The execution of an SM transaction usually only requires 
broadcasting the name of the method to be invoked; such messages are often 
shorter than 100B.
\end{itemize}

Note also that since SM transactions are guaranteed to commit, they do not 
require certification, which eliminates the certification overhead.
This overhead (in the DU mode) is proportional to the size of transactions' 
readsets and updates sets.

In \cite{KKW13} we evaluated HTR-enabled Paxos STM using manually devised 
oracles that were designed to fit the expected workload. The oracles delivered 
good performance, even though the oracles' policies were very simple: they 
either limited the abort rate, had transaction execution modes predefined for 
each transaction type or simply executed in the SM mode transactions which were 
known \emph{a priori} to cause high contention.

Naturally, the more complex the application, the more difficult designing an 
oracle which works well. Moreover, manually defined oracles have limited
capability to adjust to changing workloads. Therefore we decided to create 
mechanisms that aid the programmer in devising oracles that can adopt to
varying conditions. 

\subsection{Machine-Learning-based oracle} \label{sec:htr:ml}

Before we describe our machine learning (ML) based approach to creating 
oracles, let us first reflect on the constraints of our environment and the 
requirements that we set.

\subsubsection{Requirements and assumptions}

Determining the optimal execution mode for each transaction run (in a certain 
state of the system) can be considered a classification problem. Solving such 
problems is often accomplished by employing offline machine learning techniques 
such as decision trees, nearest neighbours or neural networks \cite{HTRF09}. 
However, it seems that resorting to such (computation-heavy) mechanisms in our 
case is not most advantageous because of the high volatility of the environment 
which we consider. Our system scarcely uses stable storage (whose performance 
is typically the limiting factor in database and distributed storage systems) 
and thus Paxos STM's performance is sensitive even to small changes in the CPU 
load. In turn, the changes could be caused by variance in one or many aspects 
of the workload such as sizes of received requests, shared object access 
patterns, request execution times, number of clients, contention levels, etc. 
Therefore, we opted for reinforcement learning techniques, i.e., approaches 
which learn by observing the \emph{rewards} on the already made decisions. 

Naturally, the primary limitation for the automated oracle is that the 
mechanism it relies on cannot incur a noticeable overhead on transaction 
processing. Otherwise, any gains resulting from choosing an optimal execution 
mode would be overshadowed by the time required for training the oracle or 
querying it. It means that we had to resort to lightweight ML techniques that 
are neither CPU nor memory intensive (see below). Also, the ML mechanism must 
work well in a multithreaded environment. This can be tricky because each query 
to the oracle is followed by a feedback on transaction execution passed to the 
ML mechanism. Note that the statistics gathered on a particular transaction run 
heavily depend on the overall load of the system, therefore calculating the 
reward (used by the ML mechanism to learn) is not straightforward.

Ideally, before a transaction is executed, the oracle should know what objects 
the transaction will access and approximately how long the execution will take. 
This is typically done in, e.g., SQL query optimizers featured in most of the 
database engines. Unfortunately, obtaining such information in our case is 
very difficult. It is because in our system transactions may contain arbitrary 
code and are specified in Java, a rich programming language, which enables 
complex constructs. One could try static code analysis as in \cite{SW12}, but 
this approach tends to be expensive and not that accurate in general case. 
However, it is reasonable to assume that not every request (transaction) 
arriving in the system is completely different from any of the already executed 
ones. Therefore, transactions can be clustered based on some easily obtainable 
information (e.g., content of the arguments passed alongside transaction's 
code), statistics on past executions that aborted due to conflicts or simple 
hints given by the programmer. The latter could range from, e.g., a qualitative 
level of contention generated by the transaction (\emph{low}, \emph{medium}, 
\emph{high}), to the number of objects accessed by the transaction compared to 
other transactions, or to as straightforward as a unique number which 
identifies a given class of transactions (as in our system, see below).

\subsubsection{Multi-armed bandit problem inspired approach}

The ML-based oracle called \emph{HybridML} (or \emph{HybML} in short), which we 
propose, relies on a rough classification provided by the programmer. As 
mentioned above, the classification may involve various elements but we 
investigate the simplest one, in which similar transactions have the same 
number 
associated with them. We say that transactions with the same number form a 
\emph{class}. For example, a class can be formed out of transactions which 
perform money transfer operations between pairs of accounts. Such transactions 
are inherently similar despite moving funds between different pairs of 
accounts. The similarities regard, e.g., shared objects access pattern, CPU 
utilization, broadcast message sizes, etc.\footnote{In case of more complicated 
transactions, which feature loops or multiple if-then statements, there might 
be bigger differences in the mentioned characteristics. Then, however, the 
programmer may easily provide slightly finer classification of the 
transactions.} 

HybML is inspired by and closely resembles the \emph{epsilon-greedy strategy} 
for solving the \emph{multi-armed bandit problem} (see \cite{R52} \cite{LR85} 
for the problem and \cite{KP14} for algorithms). In the multi-armed bandit 
problem there is a number of slot machines which, when played, return a random 
reward from a fixed but unknown probability distribution specific to that 
machine. The goal is to maximize the sum of rewards in a sequence of plays. In 
the epsilon-greedy strategy, in any given play with some small probability 
$\epsilon$ a random slot machine is chosen. In the majority of plays, however, 
the chosen machine is the one that has been performing best in the previous 
rounds. Varying the value of $\epsilon$ enables balancing of exploration and 
exploitation.

\newcommand{\edu}{\epsilon_\mathit{DU}}
\newcommand{\esm}{\epsilon_\mathit{SM}}

Roughly speaking, in HybML we use a slightly modified version of the 
ep\-si\-lon-greedy strategy to solve the two-armed bandit problem for each 
class independently (with DU and SM modes corresponding to the two slot 
machines). Firstly, HybML determines whether to optimize the network or CPU 
usage. In the former case, HybML aims at choosing an execution mode in which 
broadcast messages are smaller. Otherwise, HybML decides on an execution mode 
in which the transaction can execute and commit more quickly. 

The exact way in which HybML works is a bit more complicated. When a new 
transaction is about to start, HybML first checks what was the preferred 
execution mode for the given class of transactions. Then, depending on the 
most prevalent mode, it randomly chooses the execution mode with 
probabilities $\edu$ or $\esm$ (below we explain the reason for managing two 
values of $\epsilon$ instead of just one). Otherwise, HybML tries to optimize 
either network or CPU usage, depending on which is the observed bottleneck 
under a given workload. HybML always first ensures that network is not 
saturated, because saturating a network always results in degradation of 
performance (see \cite{WKK12} and \cite{WKK16}). To this end, HybML compares 
the values of moving averages, which store message sizes for either execution 
mode, and chooses a mode which corresponds to smaller messages. If, on the 
other hand, CPU is the limiting factor, HybML relies on moving medians, which 
store the duration of transaction execution and commit (see also the discussion 
in Section~\ref{sec:htr:ml:challenges} for the reasons on using moving medians 
instead of moving averages in case of optimizing CPU usage). Since unlike SM 
transactions, DU transactions can abort, HybML stores additional moving 
averages and medians to account for aborted DU transactions. This way, by 
knowing abort rate (measured independently for each class and accounting 
separately for conflicts detected before and after the network communication 
phase), HybML can estimate the overall cost of executing and committing a DU 
transaction (in terms of both network traffic and execution time). 

Note that the average cost of a single attempt to execute (and hopefully 
commit) a transaction in the DU mode is smaller compared to the cost of 
executing a transaction in the SM mode. It is because a DU transaction can 
abort due to a conflict and an SM transaction is guaranteed to commit. When a 
DU transaction aborts, no costly state update is performed and sometimes, if 
the conflict is detected before performing the broadcast operation, no 
resources are wasted on network communication. Therefore, in order to guarantee 
fair exploration, the probability with which the DU mode is chosen should be 
higher than the probability with which the SM mode is chosen. This observation 
led us to use $\edu$ and $\esm$ instead of a single value $\epsilon$. Currently 
$\edu = 0.01$ and $\esm = 0.1$, which could be interpreted as follows: due to a 
higher resource cost of choosing the SM mode over the DU mode, the latter is 
chosen 10 times more frequently. As shown in Section~\ref{sec:htr_evaluation}, 
the system works very well with these values, but by using abort rate, these 
values can be easily set to reflect the \emph{true} cost of an execution 
attempt.

There are few substantial differences between the definition of the original 
problem of multi-armed bandit problem and our case. Firstly, in the original 
problem the probability distributions of rewards in slot machines do not change 
and thus the strategy must account for all previous plays. HybML must be able 
to adjust to changing environment (e.g., workload) and thus it relies on moving 
medians. Most importantly, however, we treat choosing an optimal execution mode 
for any class independently, i.e., as a separate instance of the multi-armed 
bandit problem. In reality the decisions made by HybML for different classes of 
transactions are (indirectly) inter-dependent. It is because the reward 
returned after a transaction commits or aborts does not reflect solely the 
accuracy of the decision made by HybML, but it also entails the current load 
of the system. The load of the system naturally depends on all transactions 
running concurrently and thus indirectly on the decisions made by HybML for 
transactions of different classes. Note that if we were to reflect the 
inter-dependency between decisions made for different classes (in the 
form of a context as in the contextual multi-armed bandit problem 
\cite{WMP05}), the scheme would get extremely complicated and in practice it 
would never converge. 

\subsubsection{Implementation details} \label{sec:htr:ml:challenges}

Although the idea behind HybML seems simple, implementing it in a way that
it works reliably was far from easy. It is mainly because of the 
characteristics of workloads we consider in conjunction with quirks of JVM that 
we had to deal with, provided that Paxos STM is written in Java.

The biggest challenge we faced was to accurately measure the duration of 
transaction execution. In particular, we were interested in obtaining faithful
measurements on the time spent by the main thread of HTR on handling SM and DU 
transactions. When network is not saturated, the main thread becomes the 
bottleneck because it serializes execution of SM transactions with 
certification of DU transactions and is also responsible for applying 
transaction updates to the local state.\footnote{Note that parallelising 
operations in this thread does not necessarily result in better performance. It 
is because even for large transactions the cost of transaction certification or 
applying updates is comparable to the time required for completing a memory 
barrier, action which is necessary in order to preserve consistency (in the 
pseudocode, the memory barrier happens prior acquiring the lock and once it is 
released, lines \ref{alg:htr:commitBegin}, \ref{alg:htr:commitEnd}, 
\ref{alg:htr:commitBeginSM}, \ref{alg:htr:commitEndSM}; the core of Paxos STM 
uses no locks but relies on memory barriers triggered by accessing volatile 
variables).} The CPU times we measure (using the \texttt{ThreadMXBean} 
interface) are in orders of microseconds, which means that we can expect a 
large error. The instability of measurements is further amplified by the way 
Java threads are handled by JVM. In JVM, Java threads do not correspond 
directly to the low-level threads of the operating system (OS) and thus 
the same low-level OS thread which, e.g., executes transactions, can be also  
responsible for performing other tasks for the JVM such as garbage 
collecting unused objects every once in a while.\footnote{Our testing 
environment does not allow us to use low-level JNI code for enabling thread 
affinity in Java.} As a result, we often observed measurements that were 
up to $3$ orders of magnitude higher than the typical ones. As we were
unable to obtain consistent averages using relatively small windows (necessary 
to quickly adopt to changing conditions), we resorted to moving medians, 
which are less sensitive to outliers.

%% file: htr_alg_shorter.tex
\begin{algorithm*}
\caption{Hybrid Transactional Replication for process $p_i$}
\label{alg:htr}  
  \footnotesize
\begin{multicols*}{2}
\begin{algorithmic}[1]
\State{integer $\LC \gets 0$} \label{alg:htr:lc}
\State{set $\Log \gets \emptyset$} \label{alg:htr:log}

\Function{getObject}{txDescriptor $t$, objectId $\oid$}
  \If{$(\oid, \obj) \in t.\updates$} $\vvalue \gets \obj$ \label{alg:htr:retrieveUpdates}
  \State{\textbf{else} $\vvalue \gets$ retrieve object $\oid$} \label{alg:htr:retrieve}
  \EndIf
  \State{\Return $\vvalue$}
\EndFunction

\Function{certify}{integer $\start$, set $\readset$} \label{alg:htr:certify}
  \State{\LockStart $L \gets \{ {t \in \Log} : t.\eend > \start \}$ \LockEnd} 
\label{alg:htr:logFilter}
  \ForAll{$t \in L$}
      \State{$\writeset \gets \{ \oid : \exists (\oid,\obj) \in t.\updates \}$}
      \If{$\readset \cap \writeset \neq \emptyset$} \Return $\failure$
      \EndIf
  \EndFor
  \State{\Return $\success$}
\EndFunction

\Procedure{updateOracleStatistics}{txDescriptor $t$} 
\label{alg:htr:updateOracle}
    \State{$\txoracle.\textproc{feed}(t.\stats)$}
\EndProcedure

\algrule

\item[\textbf{Thread $q$ on request} $r$ \textbf{from client} $c$ (executed on 
one replica)]
\State{enum $\outcome_q \gets \failure$}\hfill\textit{// type: enum \{ success, 
failure \}}
\State{response $\res_q \gets \nnull$}
\State{txDescriptor $\tdu \gets \nnull$}\hfill\textit{// type: (id, 
start, end, readset, updates, stats)} \label{alg:htr:txDesc}

\Upon{init}{}
  \State{\textbf{wait until} $\LC \ge r.\clock$} \label{alg:htr:clock}
  \State{\Call{Transaction()}{}}
  \State{\EndReq} 
\EndUpon

\Procedure{transaction}{}
  \State{$\mode \gets \txoracle.\textproc{query}(r)$} \label{alg:htr:query}
  \If{$\mode = \mathit{DUmode}$} \label{alg:htr:dumode}
    \State{$\tdu \gets (\text{a new unique}\ \id, 0, 0, \emptyset, \emptyset, 
\emptyset)$}
    \State{\LockStart $\tdu.\start \gets \LC$ \LockEnd} \label{alg:htr:start}
    \State{$\res_q \gets$ execute $r.\prog$ with $r.\args$} \label{alg:htr:executeDU}
  \Else\hfill\textit{// mode = SMmode} \label{alg:htr:smmode}
    \State{\tobcast $r$}\hfill\textit{// blocking } \label{alg:htr:tobSM}
  \EndIf
  \If{$\outcome_q = \failure$} \Call{transaction()}{}
  \EndIf
\EndProcedure

\Function{read}{objectId $\oid$} \label{alg:htr:DUFunctionsBegin}
    \State{$\tdu.\readset \gets \tdu.\readset \cup \{\oid\}$} 
\label{alg:htr:read}
    \State{\LockStart \textbf{if} \Call{certify}{$\tdu.\start, \{\oid\}$} $= 
\failure$ \textbf{then}} \Call{retry()}{}\label{alg:htr:check}
    \State{\AlgIndent \AlgIndentSmall \textbf{else} \Return \Call{getObject}{$\tdu$, $\oid$} 
\LockEnd} \label{alg:htr:getObjectCall}
\EndFunction

\Procedure{write}{objectId $\oid$, object $\obj$}
    \State{$\tdu.\updates \gets \{ (\oid',\obj') \in \tdu.\updates : \oid' \neq 
\oid \} \cup \{(\oid,\ \obj)\}$} \label{alg:htr:write}
\EndProcedure

\Procedure{commit}{} \label{alg:htr:commit}
  \State \EndTxCode
  \If{$\tdu.\updates = \emptyset$} \label{alg:htr:noupdates}
    \State{$\outcome_q = \success$}
    \State \Return \label{alg:htr:noupdates2}
  \EndIf
  \If{\Call{certify}{$\tdu.\start, \tdu.\readset$} $= \failure$} \Return \label{alg:htr:localCert}
  \EndIf
  \State{\tobcast $\tdu$}\hfill\textit{// blocking} \label{alg:htr:tob}
  \State{\Call{updateOracleStatistics}{$\tdu$}} \label{alg:htr:updateOracleDU}
\EndProcedure

\Procedure{retry}{} \label{alg:htr:retry}
  \State{\EndTxCode}
\EndProcedure

%
\columnbreak
\Procedure{rollback}{} \label{alg:htr:rollback}
  \State{\EndTxCode}
  \State{$\outcome_q \gets \success$}
\EndProcedure \label{alg:htr:DUFunctionsEnd}

\algrule

\item[\textbf{The main thread of HTR} (executed on all replicas)]
\State{enum $\outcome \gets \nnull$}\hfill\textit{// type: enum \{ success, 
failure \}}
\State{response $\res \gets \nnull$}
\State{request $r \gets \nnull$}
\State{txDescriptor $\tsm \gets \nnull$} \label{alg:htr:txDesc2}

\Upon{\todeliver}{txDescriptor $\tdu$} \label{alg:htr:adeliverTx}
  \If{\Call{certify}{$\tdu.\start, \tdu.\readset$} $ = \success$} 
\label{alg:htr:globalCert}
     \State{\LockStart $\LC \gets \LC + 1$} \label{alg:htr:inc} \label{alg:htr:commitBegin} 
     \State{\AlgIndent $\tdu.\eend \gets \LC$ } \label{alg:htr:end} 
     \State{\AlgIndent $\Log \gets \Log \cup \{\tdu\}$} 
\label{alg:htr:logAppend}
     \State{\AlgIndent apply $\tdu.\updates$ \LockEnd} \label{alg:htr:apply}
\label{alg:htr:commitEnd}
  \If{transaction with $\tdu.\id$ executed locally by thread $q$} 
    \State{$\outcome_q \gets \success$}
  \EndIf
  \EndIf 
\EndUpon

\Upon{\todeliver}{request $r_q$} \label{alg:htr:adeliverTxSM}
  \State{$r \gets r_q$}
  \State{\Call{transaction()}{}}
  \State{\Call{updateOracleStatistics}{$\tsm$}} \label{alg:htr:updateOracleSM}
  \If{request with $r.\id$ handled locally by thread $q$}{}
    \State{$\osm \gets \outcome$}
    \State{$\res_q \gets \res$}
  \EndIf
\EndUpon

\Procedure{transaction}{}
  \State{$\tsm \gets (\text{a deterministic unique}\ \id \text{ based on 
$r.\id$}, 0, 0, \emptyset, \emptyset, \emptyset)$} \label{alg:htr:executionStartSM}
  \State{$\tsm.start \gets LC$} \label{alg:htr:startSM}
  \State{$\res \gets \nnull$}
  \State{$\res \gets$ execute $r.\prog$ with $r.\args$} \label{alg:htr:executionEndSM}
\EndProcedure

\Function{read}{objectId $\oid$, object $\obj$} \label{alg:htr:SMFunctionsBegin}
     \State{\Return \Call{getObject}{$\tsm$, $\oid$}} \label{alg:htr:readSM}
\EndFunction

\Procedure{write}{objectId $\oid$, object $\obj$}
    \State{$\tsm.\updates \gets \{ (\oid',\obj') \in \tsm.\updates : \oid' \neq 
\oid \} \cup \{(\oid,\ \obj)\}$} \label{alg:htr:writeSM}
\EndProcedure

\Procedure{commit}{} \label{alg:htr:commitSM}
   \State \EndTxCode
   \If{$\tsm.\updates \neq \emptyset$} 
    \State{\LockStart $\LC \gets \LC + 1$} \label{alg:htr:incSM} \label{alg:htr:commitBeginSM}
    \State{\AlgIndent $\tsm.\eend \gets \LC$ } \label{alg:htr:endSM} 
        \State{\AlgIndent $\Log \gets \Log \cup \{\tsm\}$} 
    \label{alg:htr:logAppendSM}
        \State{\AlgIndent apply $\tsm.\updates$ \LockEnd} \label{alg:htr:applySM} \label{alg:htr:commitEndSM}
   \EndIf
   \State{$\outcome \gets \success$}
\EndProcedure

\Procedure{retry}{} \label{alg:htr:retrySM}
  \State \EndTxCode
  \State{$\outcome \gets \failure$}
\EndProcedure

\Procedure{rollback}{} \label{alg:htr:rollbackSM}
  \State \EndTxCode
  \State{$\outcome \gets \success$}
\EndProcedure \label{alg:htr:SMFunctionsEnd}

\end{algorithmic}
\end{multicols*}
\end{algorithm*}

%% file: evaluation.tex
\section{Evaluation} \label{sec:htr_evaluation}

In this section, we present the results of the empirical study of the HTR 
scheme. To this end we compare the performance of HTR using HybML with the 
performance of HTR running with the DU or SM oracles, which execute all 
updating requests in either the DU mode or the SM mode. As we explained in 
Section~\ref{sec:htr:ml}, HybML optimizes the usage of network or CPU, 
depending on which is the current bottleneck. Since avoiding network saturation
is relatively simple, because it entails choosing the execution mode which
results in smaller messages being broadcast, we focus on the more challenging
scenario in which the processing power of the CPUs is the limiting factor. 
Under this scenario, we have to consider a much wider set of variables such as
transaction execution and commit times, contention levels, shared object access 
patterns, etc.

\subsection{Software and environment}

We conducted tests using HTR-enabled Paxos STM, our fault-tolerant object based 
DTM system written in Java which we featured in our previous work (see, e.g., 
\cite{WKK12} \cite{KKW13} \cite{WKK16}). Paxos STM relies on a fast 
implementation of TOB based on Paxos \cite{Lam98} and implements optimizations 
such as multiversioning and early conflict detection (see also Sections 
\ref{sec:related_work:dtm} and \ref{sec:context:comp}).

We run Paxos STM in a cluster of 20 nodes connected via 16Gb Ethernet 
over Infiniband. Each node had 28-core Intel E5-2697 v3 2.60GHz processor 64GB 
RAM and was running Scientific Linux CERN 6.7 with Java HotSpot 1.8.0.

\subsection{Benchmarks} \label{sec:htr_evaluation:benchmarks}

In order to test the HTR scheme, we extended the hashtable microbenchmark, 
which we used in \cite{WKK12} and \cite{WKK16}. The benchmark features a 
hashtable of size $h$, storing pairs of key and value 
accessed using the \emph{get}, \emph{put}, and \emph{remove} operations. A run 
of this benchmark consists of a load of requests (transactions) which are 
issued to the hashtable, each consisting of a series of \emph{get} operations 
on a randomly chosen keys and then a series of update operations (either 
\emph{put} or \emph{remove}). Initially, the hashtable is prepopulated with 
$\frac{h}{2}$ random integer values from a defined range, thus giving the 
saturation of 50\%. This saturation level is always preserved: if a randomly 
chosen key points at an empty element, a new value is inserted; otherwise, the 
element is removed. 

In the current implementation of the benchmark, we can adjust several 
parameters for each class of transactions independently and at run-time. The 
parameters include, among others, the number of read and write operations, the 
subrange of the hashmap from which the keys are chosen, access 
pattern (random keys or a continuous range of keys) and the duration of the 
additional sleep operation, which is invoked during transaction execution in 
order to simulate computation heavy workload. By varying these parameters and 
the ratio of concurrently executing transactions of different classes, we can 
generate diverse workloads, which differ in CPU and network usage and are 
characterised by changing contention levels.

We consider three test scenarios: \emph{Simple}, \emph{Complex} and 
\emph{Complex-Live}. We use the first two scenarios to evaluate the throughput 
(measured in requests per second) and scalability of our system. The latter 
scenario is essentially the Complex scenario, whose parameters are changed 
several times throughout the test. We use the Complex-Live scenario to 
demonstrate the ability of HybML to adjust to changing conditions at run-time 
(we show the throughput of HTR in the function of time). For each scenario we 
define from 2 up to 11 classes of transactions, whose parameters are 
summarized in Figure~\ref{tab:htr_evaluation:transactions}). We have chosen the 
parameters so that one can observe the strong and weak aspects of HTR running 
with either the DU or SM oracle. This way we can demonstrate HybML's ability to 
adapt to different conditions. 
In order to utilize the processing power of the system across different 
cluster configurations, we increase the number of requests concurrently 
submitted to the system with the increasing number of replicas.

\begin{figure*}
\def\arraystretch{1.2}
\center
\scalebox{0.75}{
\begin{tabular}{ r r r || r | r r r r r r r r r r r }
\multicolumn{3}{r||}{Scenario} & Parameter & $T_0$ & $T_1$ & $T_2$ & $T_3$ & 
$T_4$ & $T_5$ & $T_6$ & $T_7$ & 
$T_8$ & $T_9$ & $T_{10}$ \\ 
\hline
\multicolumn{3}{r||}{\multirow{4}{*}{Simple}} & probability & 90 & 10 & & & & & 
& & & & \\
& & & reads & 2500 & 300 & & & & & & & & & \\
& & & updates & 0 & 5 & & & & & & & & & \\
& & & range & 600k & 600k & & & & & & & & & \\
\hline
\multicolumn{3}{r||}{\multirow{5}{*}{Complex}} & probability & 90 & 1 & 1 & 1 & 
1 & 1 & 1 & 1 & 1 & 1 & 1 \\
& & & reads & 2500 & 200 & 200 & 200 & 200 & 200 & 200 & 200 & 200 & 200 & 
200 \\
& & & updates & 0 & 5 & 5 & 5 & 5 & 5 & 5 & 5 & 5 & 5 & 5 \\
& & & range & 10.24M & 5.12M & 2.5M & 1.28M & 640k & 320k & 160k & 80k & 40k & 
20k & 10k \\
& & & offset & 0 & \multicolumn{10}{c}{Distinct part of Hashtable for each 
$T_1$--$T_{10}$}\\
\hline
\multirow{6}{*}{Complex-Live} & a & (0-200 s) & & & \multicolumn{10}{c}{As in 
Complex}\\
\cline{2-15}
& \multirow{2}{*}{b} & \multirow{2}{*}{(200-400 s)} & reads & 2500 & 
\multicolumn{10}{c}{400 (2x as in Complex)}\\
& & & updates & 0 &  \multicolumn{10}{c}{10 (2x as in Complex)}\\
\cline{2-15}
& c & (400-600 s)& sleep & 0 &  \multicolumn{10}{c}{Extra 0.1ms sleep in every 
transaction}\\
\cline{2-15}
& d & (600-800 s) & range & 10.24M & \multicolumn{10}{c}{Half the range from 
Complex for each $T_1$--$T_{10}$}\\
\cline{2-15}
& e & (800-1000 s)& & & \multicolumn{10}{c}{As in Complex}\\
\end{tabular}}
\caption{Benchmark parameters for different test scenarios. In the 
$\emph{Complex-Live b-d}$ scenarios we change some parameters compared to the 
\emph{Complex} scenario.}
\label{tab:htr_evaluation:transactions}
\end{figure*}

\subsection{Benchmark results}

In Figures \ref{fig:htr_evaluation:results}, 
\ref{fig:htr_evaluation:results_random} and 
\ref{fig:htr_evaluation:live_results} we present the test results of HTR. Below 
we discuss the test results in detail.

\subsubsection{The Simple Scenario}

In this scenario, for which the test results are given in 
Figure~\ref{fig:htr_evaluation:results}a, there are only two 
classes of transactions ($T_0$ and $T_1$), which operate on a hashmap of size 
$h$=600k. $T_0$ transactions (i.e., transactions, which belong to the $T_0$ 
class) are read-only and execute 2500 read operations in each run. $T_1$ 
transactions perform 300 read and 5 updating operations. The ratio between 
transactions $T_0$ and $T_1$ is 90:10. 

In this scenario the throughput of HTR running with the SM oracle remains 
constant across all cluster configurations. It is because the execution of all 
$T_1$ transactions needs to be serialized in the main thread of HTR, which 
quickly becomes the bottleneck. The throughput of 150k tps (transactions per 
second), achieved already for 4 nodes, indicates the limit on the number of 
transactions that the system can handle in any given moment. Also, the 
transactions executed in the SM mode never abort, thus the abort rate is zero. 

In the case of HTR running with the DU oracle, with the increasing number of 
replicas, the throughput first increases, then, after reaching maximum for 5 
replicas, slowly diminishes. The initial scaling of performance can be 
attributed to increasing processing power that comes with a higher number of 
replicas taking part in the computation. In the 5 node configuration, the peak 
performance of HTR running with the DU oracle is achieved. As in case of the SM 
oracle, the main thread of HTR becomes saturated and cannot process any more 
messages which carry state updates. Naturally, with the increasing number of 
concurrently executed transactions, one can observe the raising number of 
transactions aborted due to conflicts. Therefore, adding more replicas results 
in diminishing performance. The abort rate of almost 40\% in the 20 node 
cluster configuration means that every updating transaction is on average 
executed 7.5 times before it eventually commits. 

The HybML oracle takes advantage of the scaling capabilities of DUR for smaller 
cluster configurations--the performance yielded by HybML is on par with the 
performance of the DU oracle, because HybML always chooses the DU mode for all 
updating transactions. From 10 nodes upwards, the performance of the DU oracle 
drops below the performance of the SM oracle. The 10-12 node configurations are 
problematic for the HybML oracle, because the relative difference in 
performance of DU and SM modes is modest, and thus it is difficult for the 
oracle to make an optimal decision. This is why we can observe that HybML still 
chooses the DU mode for the $T_1$ transactions, instead of SM. Then, however, 
the differences start to increase, and HybML begins to favour the SM mode over 
the DU mode for the $T_1$ transactions (see the third diagram in 
Figure~\ref{fig:htr_evaluation:results}a, which shows the dominant execution 
mode for each class in HybML; different colours signify the relative ratio 
between executions in the DU and the SM modes). Eventually (from the 15 node 
configuration upwards), HybML always chooses the SM mode thus yielding the same 
performance as the SM oracle.

\subsubsection{The Complex Scenario}

% ...
In the Complex scenario, for which the evaluation results are 
given in Figure~\ref{fig:htr_evaluation:results}b, there is only one class of 
read-only transactions ($T_0$, the same as in the Simple scenario) and 10 
classes of updating transactions ($T_1$-$T_{10}$). Each updating transaction 
has the same likelihood of being chosen and each performs 200 read and 5 
update operations. However, for each class we assign a disjoint subrange 
of the hashmap. It means that each (updating) transaction from a given class 
can only conflict with transactions, which belong to the same class. Because 
the sizes of subranges are different for every class of transaction, the 
contention levels for each class will greatly differ: they would be lowest 
for transaction $T_1$ (whose range encompasses 5.12M keys) and highest for 
$T_{10}$ (whose range encompasses just 10k keys). The exact values of the 
subranges are given in Figure~\ref{tab:htr_evaluation:transactions}).


As in case of the Simple scenario, in the Complex scenario the DU oracle 
first yields better performance than the SM oracle. The highest throughput of 
about 200k tps is achieved by the DU oracle for the 8 nodes configuration, 
while the performance of the SM oracle levels at about 160k tps. Note that for 
the 3-5 nodes configuration, the SM oracle performance scales. It is because 
in the Complex scenario the updating transactions are shorter than in 
the Simple scenario. Therefore, in order to saturate the main thread of HTR, 
more concurrently submitted requests are needed (the thread becomes saturated 
in the 5 nodes configuration). 

For larger cluster configurations, the performance of the DU oracle degrades 
due to the rising number of conflicts. As a result, the performance of the DU 
oracle drops below the performance of the SM oracle for the 15 nodes 
configuration.


Note that the abort rate levels, which we can observe for the DU oracle, are 
very similar to the ones we saw in the Simple scenario. However, there are 
significant differences between the relative abort rates measured for each 
transaction class independently. For instance, for the 20 node configuration, 
the abort rate is about 4\% for $T_1$ and over 99\% for $T_{10}$ (in the latter 
case a transaction is on average aborted 150 times before it eventually 
commits).

In this scenario, HybML demonstrates its ability to adjust to the workload and 
achieves performance that is up to 40\% higher than the DU oracle's and up to 
75\% higher than the SM oracle's. This impressive improvement in performance 
justifies our ML-based approach. The plot show that HybML maintains a 
relatively low abort rate of about 7-8\% across all cluster configurations. One 
can see that the higher number of concurrently executed transactions, the higher
percentage of transactions is executed by HybML in the SM mode thus keeping 
contention levels low. Naturally, HybML chooses the SM mode first for 
the $T_{10}$ transactions, for which the contention level is the highest. Then, 
gradually, HybML chooses the SM mode also for transactions, which belong to 
classes $T_9$, $T_8$ and also $T_7$. For other transactions the cost of 
execution in the DU mode is still lower than the cost of execution in the SM 
mode, and thus HybML always chooses for these transactions the DU mode, 
regardless of the cluster configuration.

\newcommand{\plotscale}{0.38}

\begin{figure}
\hspace{-0.50cm}
\begin{tabular}{p{1.65in} p{1.6in}}
\small
\hspace{1.85cm} a) Simple & \hspace{1.6cm} b) Complex\\
\includegraphics [scale=\plotscale] 
{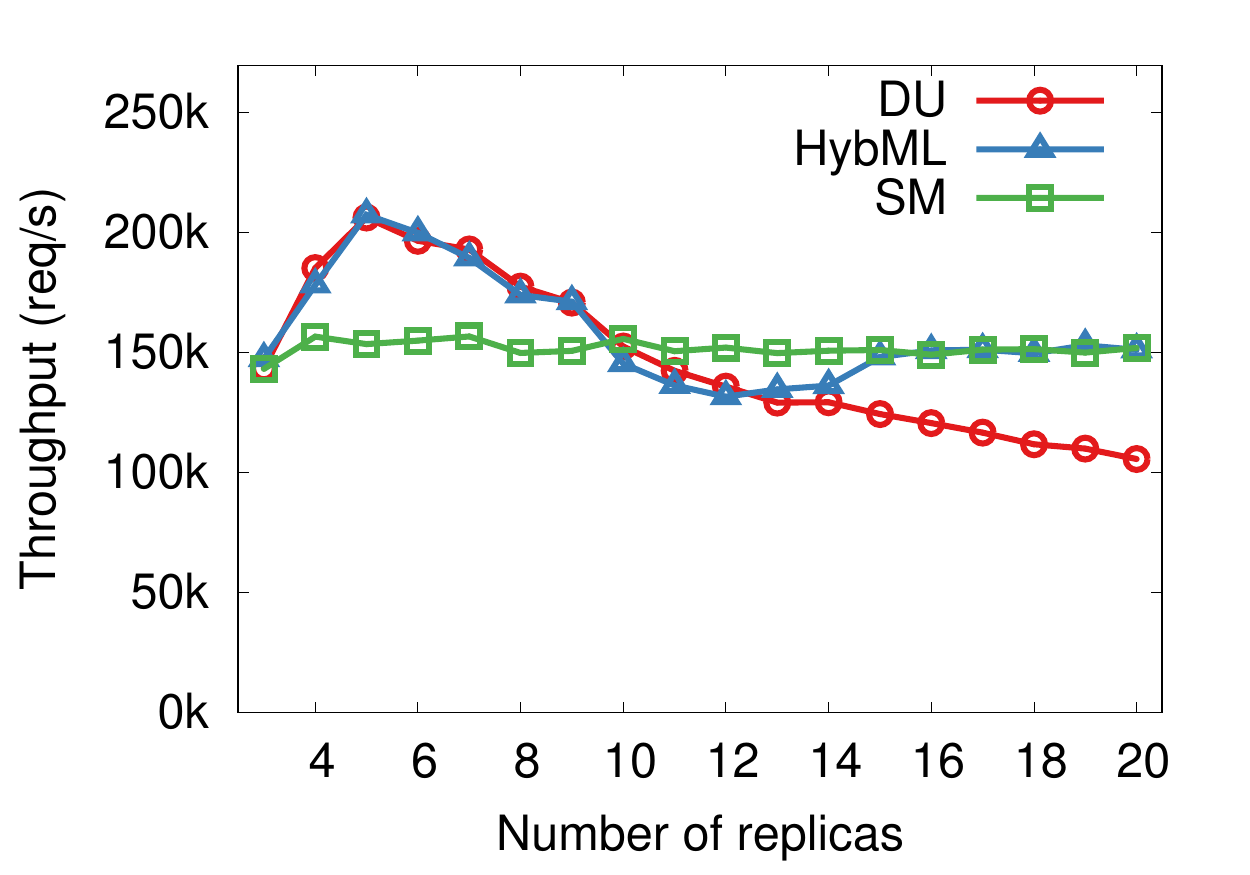} &
\includegraphics [scale=\plotscale] 
{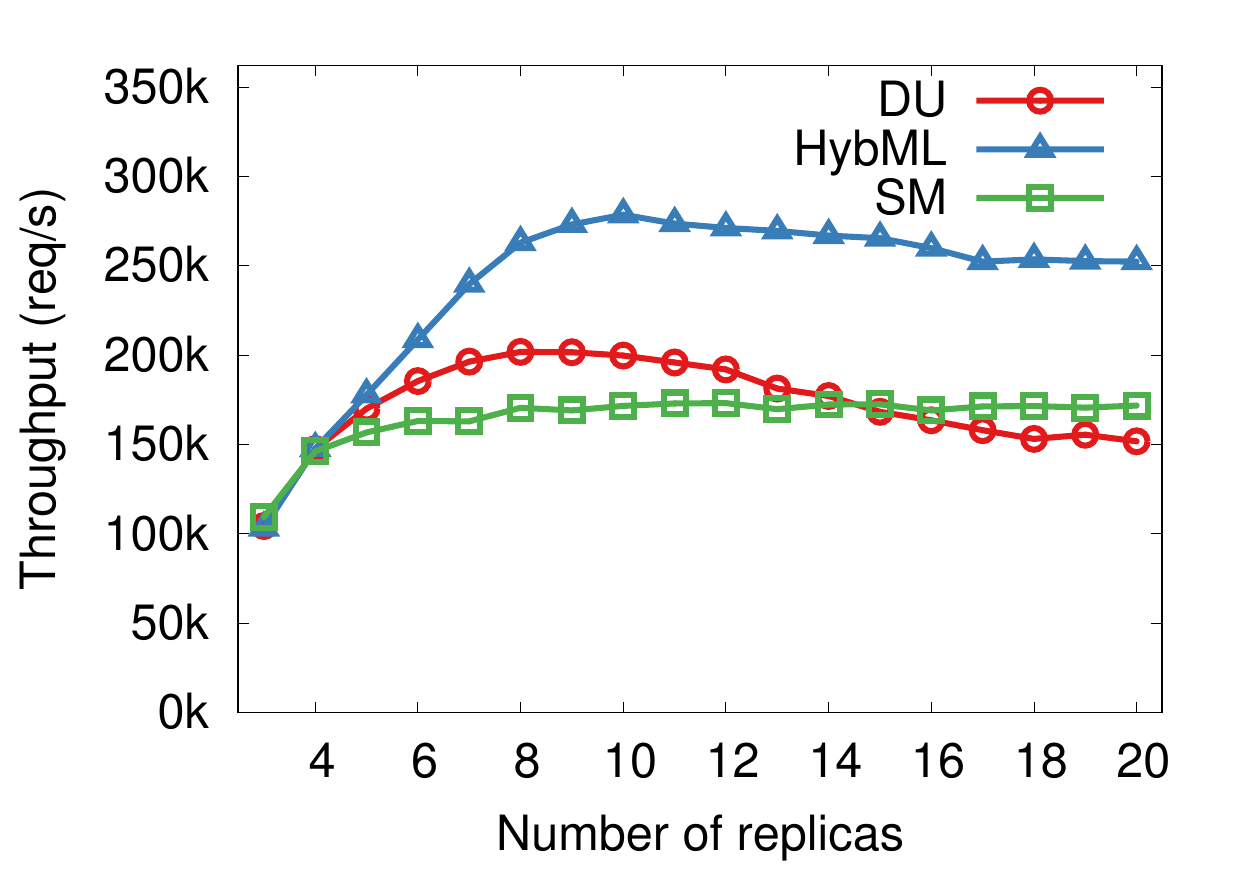} \\ [-6pt]
\hspace*{0.05cm}\includegraphics [scale=\plotscale] 
{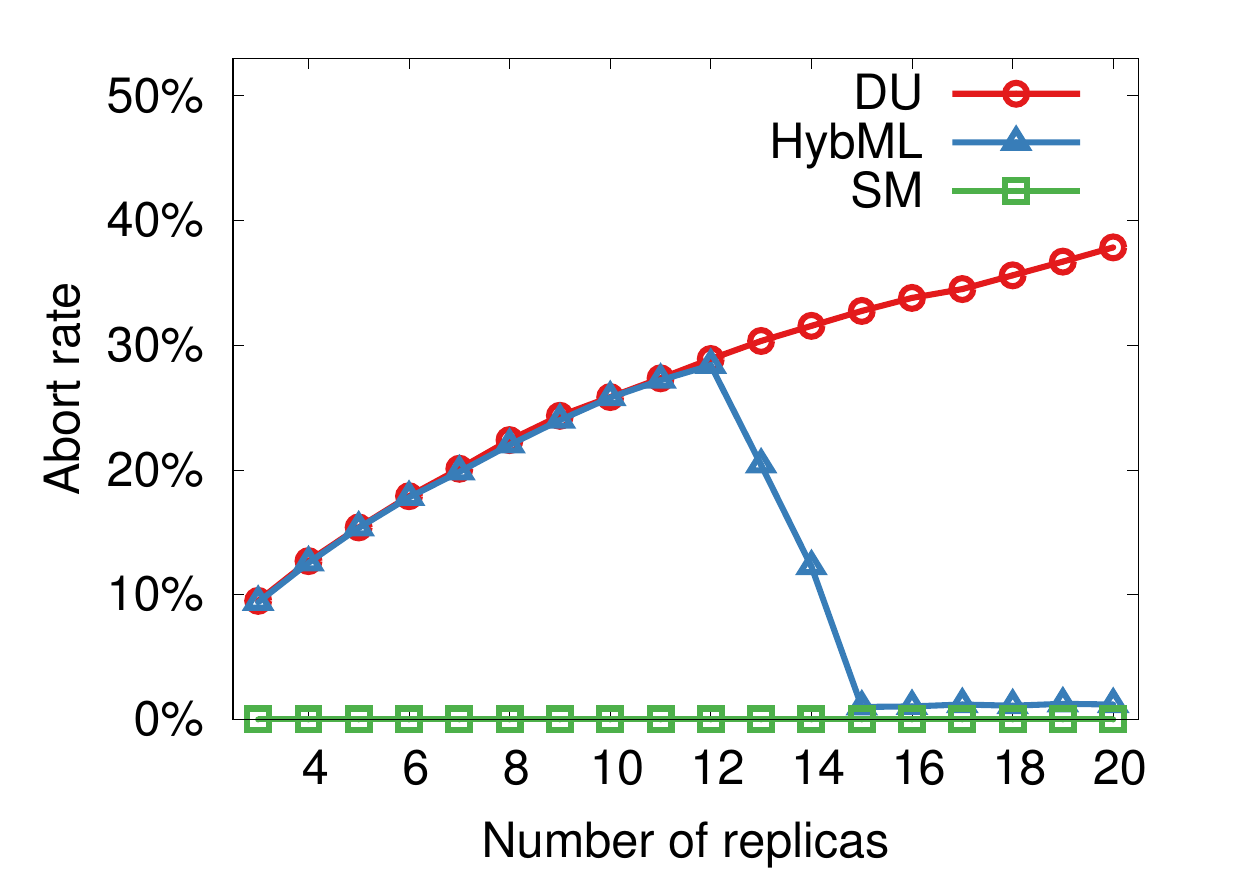} &
\hspace*{0.05cm}\includegraphics [scale=\plotscale] 
{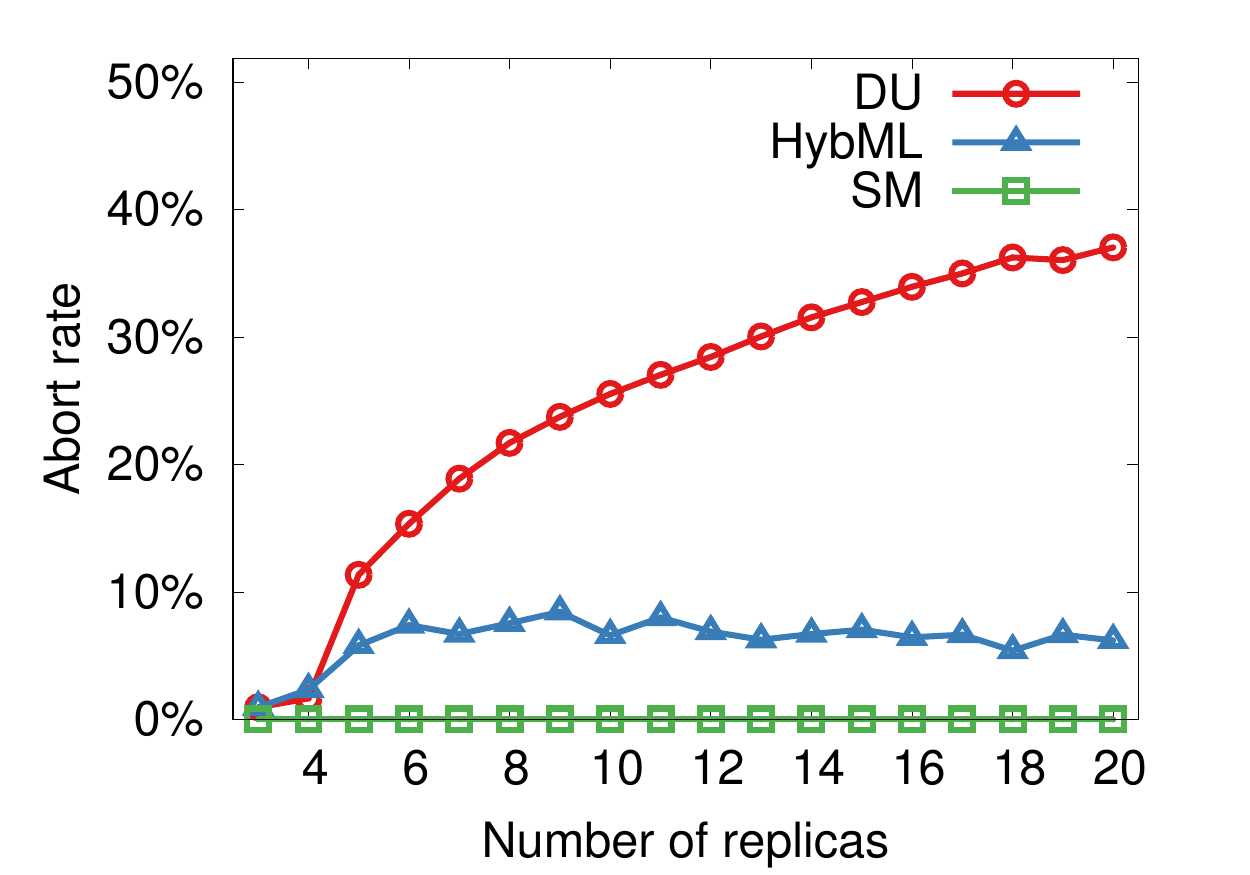} \\ [-6pt]
\vspace{-0.25cm}\hspace*{-0.35cm}\includegraphics[scale=\plotscale, angle=270] 
{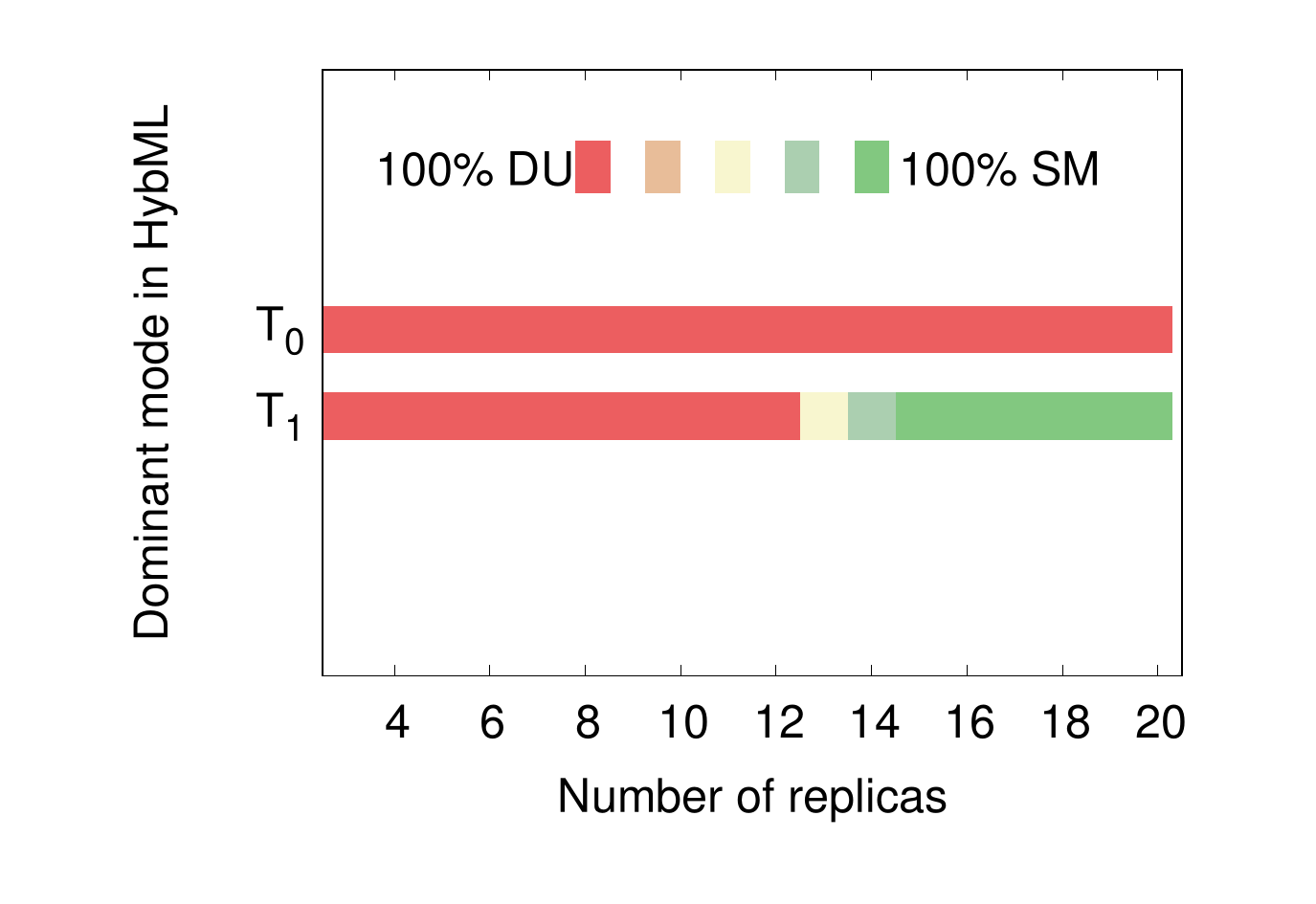} &
\vspace{-0.25cm}\hspace*{-0.28cm}\includegraphics[scale=\plotscale, angle=270] 
{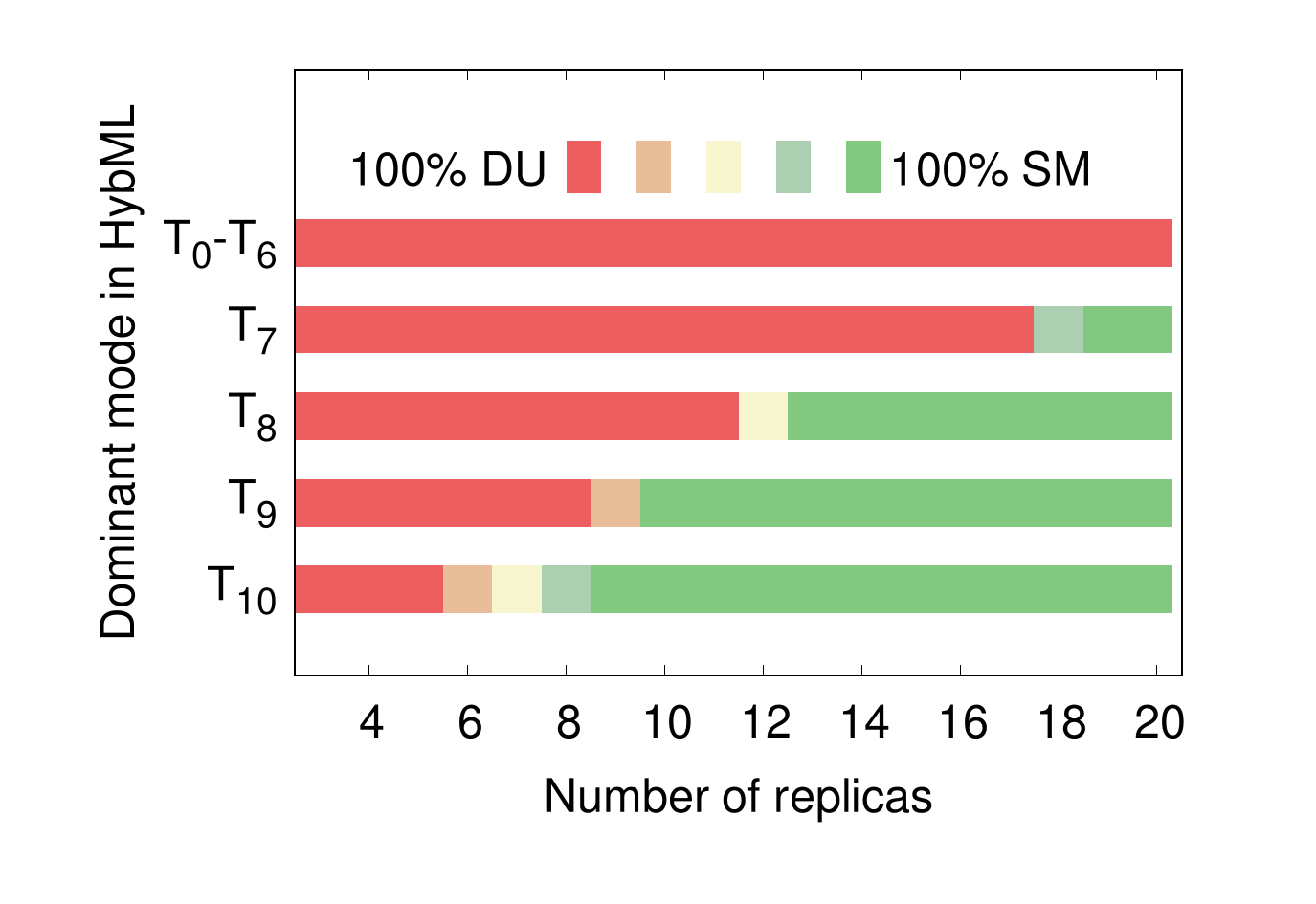}
\end{tabular}
\caption{The performance of HTR with different oracles across various cluster 
configurations.} 
\label{fig:htr_evaluation:results}
\end{figure}

In order to check the consequences of providing HybML with an inaccurate 
classification of transactions, we purposefully marked some percentage
of updating transactions with a random number corresponding to some other class.
The results of this experiment are given in 
Figure~\ref{fig:htr_evaluation:results_random}. Naturally, as a baseline we 
used the performance of HybML from the previous test. Understandably, with 10\% 
or 30\% of incorrectly marked transactions (\emph{HybML 10\% error} and 
\emph{HybML 30\% error} in the Figure), HybML still performs better than either 
the DU or SM oracles but not as fast as previously (the performance for the 
10\% and 30\% mistake scenarios peaked at 240k tps and 200k tps, respectively).
This result indicates that HybML gracefully handles even quite significant 
errors in the classification provided by the programmer.

\begin{figure}
\center
\includegraphics [scale=0.5] 
{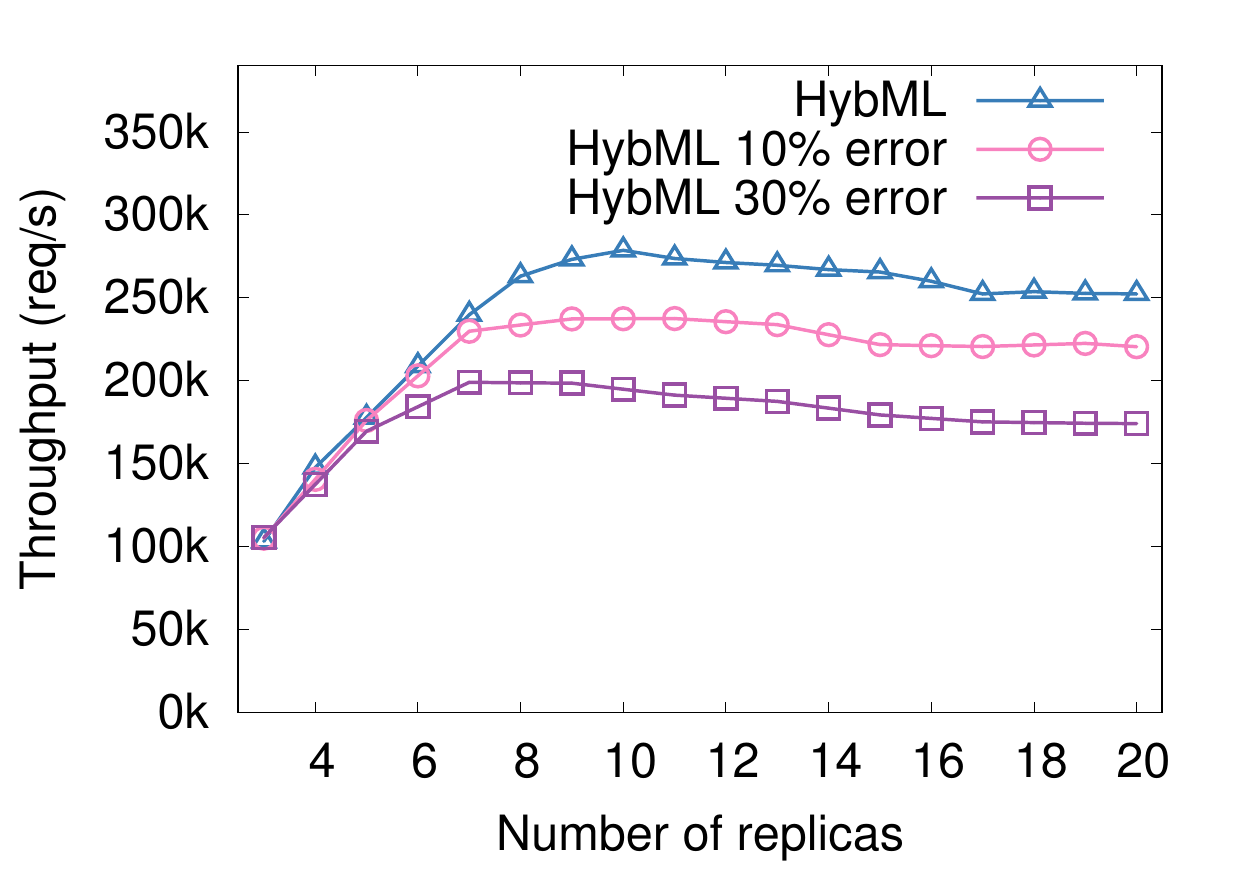}
\caption{The performance of HybML in the Complex scenario, when HybML
is provided with an inaccurate classification of transactions.} 
\label{fig:htr_evaluation:results_random}
\end{figure}

\subsubsection{The Complex-Live scenario}

\begin{figure}
\includegraphics [scale=0.7]{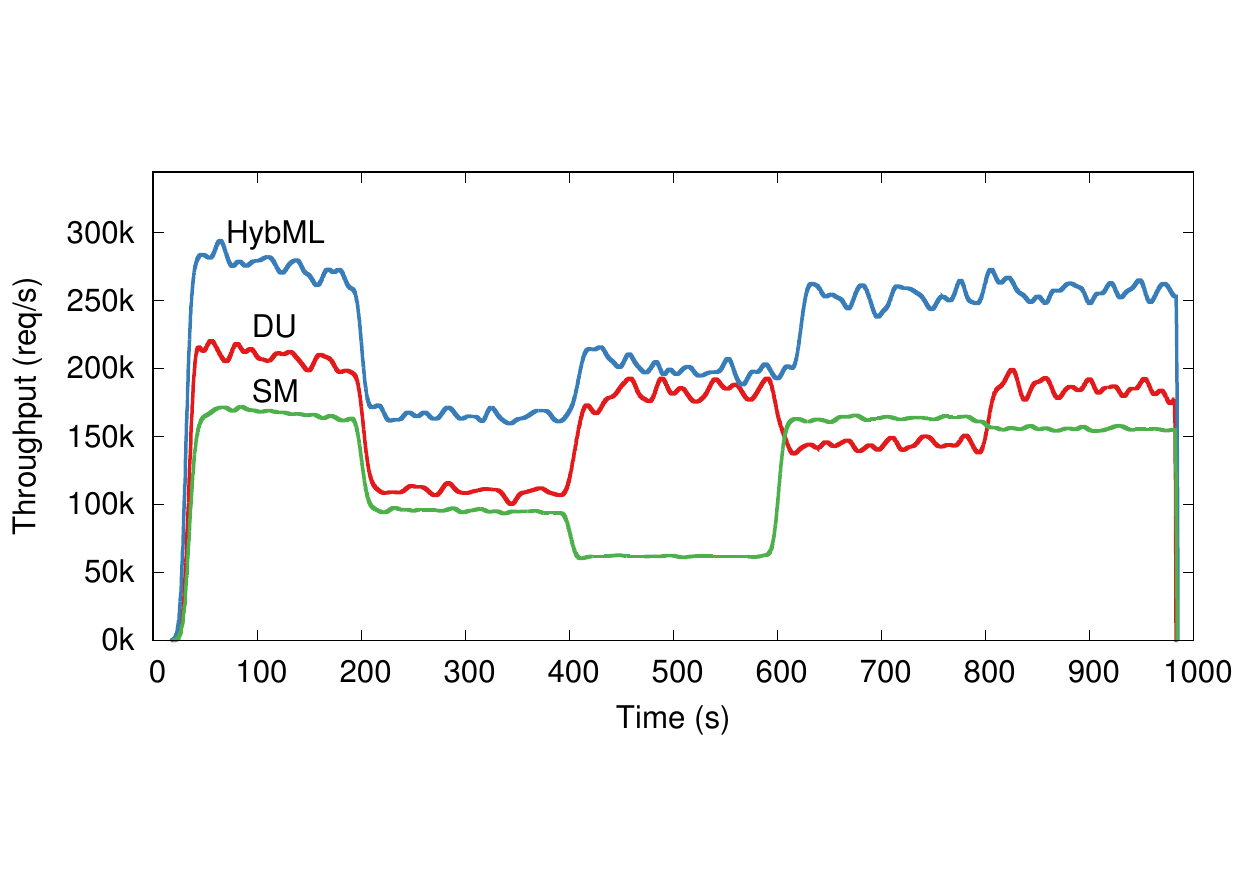}
\caption{The Complex-Live scenario: the performance of HTR in the function of 
time.} 
\label{fig:htr_evaluation:live_results}
\end{figure}

In the Complex-Live scenario we demonstrate the ability of HybML to adapt in 
real-time to changing conditions. To this end we consider the system consisting 
of 9 nodes and a workload identical with the one from the Complex scenario, 
which we then change several times during a 1000 seconds run. A plot showing 
the throughput of HTR with different oracles is given in 
Figure~\ref{fig:htr_evaluation:live_results}. One can see that in all cases the 
HybML oracle gives better performance than either the DU or SM oracles and 
almost instantly reacts to changes of the workload.

During the first 200 seconds the observed performance matches the results from 
the Complex scenario. The throughput fluctuates a little bit because of the 
garbage collector, which periodically removes unused objects from  
memory.\footnote{Executing a transaction in the DU mode results in more 
noticeable overhead due to the garbage collector. Hence, one can observe bigger 
fluctuations in throughput for the DU and HybML oracles compared to the SM 
oracle.} Towards the 200th second the throughput slightly decreases as garbage 
collecting becomes regular. 

In the 200th second we change the parameters of the benchmark, so now each 
updating transaction performs twice the number of read and updating operations 
as before (see \emph{Complex-Live b} scenario in 
Figure~\ref{tab:htr_evaluation:transactions}). The performance of the system 
decreases, because such a change results in longer transaction execution times 
and larger messages. HybML performs 70\% better than the SM oracle and over 
50\% better than the DU oracle. 

Between the 400th and 600th second, the benchmark parameters are the same as in 
the Complex scenario but the execution of each updating transaction is 
prolonged with 0.1 ms sleep thus simulating a computation heavy workload. 
Naturally, such workload is troublesome for the SM oracle, because all updating 
transactions are executed sequentially. The additional 0.1 ms sleep is 
handled well by the system when transactions are executed in the DU mode, 
because transaction execute in parallel. HybML achieves about 15\% better 
performance than the DU oracle, as it allows about 75\% of the $T_{10}$ 
transactions (which are most likely to be aborted due to conflicts) to be 
executed in the SM mode, thus reducing the abort rate and saving on transaction 
reexecutions. 

The change to the benchmark parameters in the 600th second involves reducing 
two times the size of the hashmap subrange for each class. This way the 
updating transactions, which perform the same number of read and updating 
operations as in the Complex scenario, are much more likely to abort due to 
conflicts. This change is reflected by a steep decrease in performance of the 
DU oracle. On the other hand, the performance of the SM oracle is almost the 
same as in the Complex scenario, because execution of all updating transactions 
takes the same amount of work as in the Complex scenario. Smaller subranges 
impact only data locality, which is now better and thus translates into a 
slightly better performance. Stunningly, the performance of HybML is almost the 
same as in the Complex scenario: HybML automatically started to execute a 
higher percentage of updating transactions in the SM mode thus keeping the 
abort rate low. The achieved throughput is over 65\% better than with the DU 
oracle and over 50\% than with the SM oracle. 

The last 200 seconds of the test is performed with the parameters from the 
Complex scenario. HybML quickly relearns the workload and 
starts to perform as in the first 200s of the test.

\subsection{Evaluation Summary}

We tested the HTR scheme with three oracles: DU, SM and HybML. The DU and SM 
oracles execute all updating transactions either in the DU or SM mode. 
Therefore, a system using these oracles resembles an implementation of DUR (see 
Section~\ref{sec:context:dur}) and an implementation of the optimized version 
of SMR, which allows read-only requests to be executed in parallel. Unlike the 
DU and SM oracles, the HybML oracle mixes transaction execution modes to 
achieve better performance and scalability (as evidenced by 
Figure~\ref{fig:htr_evaluation:results}). Our tests show that 
HybML provides performance that is at least as good as with either DU or SM 
(when the difference in performance between the system running with the DU or 
SM oracles is large enough) and often exceeds it by up to 50-70\% across a wide 
range of cluster configurations and types of workload. HybML avoids the 
pitfalls of either SMR and DUR and handles very well the workloads that are 
notoriously problematic for either replication scheme (i.e., computation 
intensive workloads in SMR and workloads characterised by high contention 
levels in DUR). 

We also demonstrated HybML's ability to quickly adjust to changing conditions. 
The automatic adaptation to a new workload type happens smoothly and almost 
instantly, without even temporary degradation of performance, compared to the 
performance under stable conditions (see 
Figure~\ref{fig:htr_evaluation:live_results}). All the benefits of the HTR 
scheme running with the HybML oracle require only minimal input from the 
programmer, which involves providing a rough classification of transactions 
submitted to the system. Slight inaccuracies in the classification do not 
heavily impact the performance achieved by HTR running with HybML.

%% file: conclusions.tex
\section{Conclusions} \label{sec:conc}

In this paper, we presented and evaluated Hybrid Transactional Replication, a 
novel scheme for replication of services. The two transaction execution modes 
that are used in HTR, i.e., deferred update and state machine, complement each 
other. The DU mode allows for parallelism in transaction execution, while the 
SM mode provides abort-free transactions which are useful to deal with 
irrevocable operations and transactions generating high contention. Dynamic 
switching between the modes enables HTR to perform well under a wide range of 
workloads, which is not possible for either of the schemes independently.  

The test results indicate the viability of our ML-based approach to determining
an optimal execution mode for each transaction run. Not only HTR with ML-based 
oracles achieves good performance under various workloads, but it can also 
dynamically adapt to changing conditions. This desirable behaviour of HTR does 
not come at the cost of weaker guarantees for clients: as we formally prove, 
HTR offers strong consistency guarantees akin to those provided by other 
popular transactional replication schemes such as Deferred Update Replication. 
This makes HTR a truly versatile solution.

%% file: appendix.tex
This supplemental material is an appendix of the paper: \emph{Hybrid 
Replication: State-Machine-based and Deferred-Update Replication Schemes
Combined}, containing the proofs of lemmas and theorems. See the manuscript and 
\cite{KKW16} for the definition of terms and symbols that appear in the 
proofs below.

%
%
%
%
%
%

\section{The correctness of HTR} \label{sec:app:htr}

\newcommand{\resp}{\mathit{resp}}
\newcommand{\rread}{\mathit{read}}
\newcommand{\wwrite}{\mathit{write}}
\newcommand{\ok}{\mathit{ok}}

\newcommand{\DRO}{\mathit{DRO}}
\newcommand{\texec}{\mathit{texec}}
\newcommand{\tryC}{\mathit{tryC}}
\newcommand{\tryA}{\mathit{tryA}}
\newcommand{\true}{\mathit{true}}
\newcommand{\visible}{\mathit{visible}}

\newcommand{\abort}{\mathit{abort}}
\newcommand{\alg}{\mathit{prog}}
\newcommand{\context}{\mathit{context}}
\newcommand{\execute}{\textsc{execute}}
\newcommand{\perform}{\mathit{perform}}
\newcommand{\step}{\mathit{step}}
\newcommand{\steps}{\mathit{steps}}
\newcommand{\vis}{\mathit{vis}}

\newcommand{\dumode}{\mathit{DUmode}}
\newcommand{\smmode}{\mathit{SMmode}}

Below we consider only t-histories of HTR, i.e., histories limited
to events that are related to operations on t-objects ($\texec$ operations) 
and controlling the flow of transactions such as \emph{commit} and \emph{abort} 
events ($\tryC$ and $\tryA$ operations, respectively). In this sense, we treat 
the implementation of HTR as some TM object $M$, and reason about t-histories 
$H|M$.


\htrno*

\begin{proof}
Trivially, every t-history of DUR is also a valid t-history of HTR 
(transactions in DUR are handled exactly the same as DU transactions in HTR. 
Since DUR does not satisfy write-real-time opacity \cite{KKW16}, neither 
does HTR.
\end{proof}


\htrnoro*

\begin{proof}
The proof follows directly from Theorem~\ref{thm:htr_no} and definitions of 
write-real-time opacity and real-time opacity (real-time opacity is strictly 
stronger than write-real-time opacity).
\end{proof}


In the following propositions by state of some process $p_i$ we understand the 
combined state of all t-objects maintained by $p_i$ and the current values of 
$\LC$ and $\Log$ that $p_i$ holds (but excluding statistics held in transaction 
descriptors, which do not count as part of the state).

The following proofs in many places are analogous to the proofs from of 
\cite{KKW16}, where we showed that DUR satisfies update-real-time opacity. 

\begin{proposition} \label{prop:htr:noconcurrency}
Let $k_a \cdot k_b$ and $k'_a \cdot k_b$ be such that:
\begin{enumerate}
\item $k_a$ is a certification of a DU transaction whose transaction descriptor 
has been delivered using TOB and $k_b$ is modifying the system's state 
afterwards,
\item $k'_a$ is an execution of an SM transaction and $k_b$ is modifying the 
system's state afterwards. 
\end{enumerate}
Let $K_1$ be either $k_a \cdot k_b$ or $k'_a \cdot k_b$, and $K_2$ also be 
either $k_a \cdot k_b$ or $k'_a \cdot k_b$ but $K_1$ and $K_2$ pertain to 
different transactions. For any process $p_i$ executing HTR, $K_1$ and $K_2$ 
never interleave, and changes to the state of $p_i$ happen atomically only 
after $k_b$.
\end{proposition}

\begin{proof}
The state of any $p_i$ changes only if the value of $\LC$, $\Log$ or any 
t-object changes ($k_b$). This can happen only when $p_i$ delivers a 
message through TOB, i.e., either when $p_i$ processes a transaction descriptor 
of a DU transaction ($k_a \cdot k_b$, lines 
\ref{alg:htr:commitBegin}--\ref{alg:htr:commitEnd}) or 
when $p_i$ processes a request which then $p_i$ executes as an SM transaction 
($k'_a \cdot k_b$, lines 
\ref{alg:htr:commitBeginSM}--\ref{alg:htr:commitEndSM}). $p_i$ can process 
only one message at a time (these messages are processed as non-preemptable 
events). Therefore, both $K_1$ and $K_2$ happen atomically and sequentially to 
each other. 
\end{proof}

\begin{proposition} \label{prop:htr:deterministicdu}
Let $p_i$ be a process executing HTR. Let $t$ be a transaction descriptor 
of a DU transaction delivered using TOB by $p_i$ and let $S$ be the state of 
$p_i$ at the moment of delivery. Let $S'$ be the state of $p_i$ after $p_i$ 
certifies and (possibly) updates its state, $\Log_i$ be the log of $p_i$ in 
state $S'$ and $t_i$ be the value of $t$ such that $t_i \in \Log_i$ in case of 
successful certification of the transaction. Then for every process $p_j$ in 
state $S$, if $p_j$ delivers $t$ using TOB, then $p_j$ moves to state $S'$.
\end{proposition}

\begin{proof}
By Proposition~\ref{prop:htr:noconcurrency}, the state of $p_k$ does not change 
throughout certification of a DU transaction (whose transaction descriptor has 
been delivered) and applying the updates produced by the transaction.

Since the certification procedure (line \ref{alg:htr:certify}) is deterministic 
and the values of the $\Log$ variables are equal between processes (except for 
the statistics field which is not used by the procedure), the procedure yields 
the same result. If the outcome is negative, neither process changes its state 
(line \ref{alg:htr:globalCert}). Otherwise, both processes increment $\LC$ to 
the same value (line \ref{alg:htr:inc}), assign $\LC$'s current value to the 
$\eend$ field of the transaction descriptors (line \ref{alg:htr:end}, the value 
of $\LC$ could not change during processing of the transaction descriptor). 
Next, processes append the transaction descriptors to the $\Log$ (line 
\ref{alg:htr:logAppend}) and then apply $t.\updates$ (line 
\ref{alg:htr:apply}). Therefore both processes move to the same state $S'$.
\end{proof}

\begin{proposition} \label{prop:htr:deterministicsm}
Let $p_i$ be a process executing HTR. Let $r$ be a request delivered 
using TOB by $p_i$, let $S$ be the state of $p_i$  at the moment of 
delivery of $r$ and $S'$ be the state of $p_i$ after execution of $r$ as an SM 
transaction $T_k$ with transaction descriptor $t_k$. For every process $p_j$ in 
state $S$, if $p_j$ delivers $r$ using TOB then execution of $r$ as an SM 
transaction $T_l$ (with transaction descriptor $t_l$) by $p_j$ yields state 
$S'$ of $p_j$ and $t_k = t_l$ (except for the statistics field).
\end{proposition}

\begin{proof}
By Proposition~\ref{prop:htr:noconcurrency}, the state of $p_k$ does not change 
throughout execution of an SM transaction and applying the updates produced
by the transaction.

The values of $t_k.\id$ and $t_l.\id$ are equal, since processes assign to the 
$\id$ field a value which deterministically depends on $r.\id$ (line 
\ref{alg:htr:executionStartSM}). 

Since both $p_i$ and $p_j$ start execution of $r$ from the same state, the 
current values of their $\LC$ variables are equal. Hence, $t_k.\start = 
t_l.\start$ (line \ref{alg:htr:startSM}).

During execution of an SM transaction nothing is ever added to $\readset$ 
(line \ref{alg:htr:readSM}). Therefore $t_k.\readset = t_l.\readset = 
\emptyset$.

Since HTR assumes that only a request with deterministic $\prog$ can be 
executed as an SM transaction, $r.\prog$ must be deterministic. Both processes 
execute $r.\prog$ with the same $r.\args$ (line \ref{alg:htr:executionEndSM}) 
and operate on the same state $S$ which does not change throughout the 
execution of $r.\prog$. Moreover, all updates produced by the transactions are 
stored in the $\updates$ sets (line \ref{alg:htr:writeSM}). Therefore 
$t_l.\updates = t_k.\updates$.

Also $t_k.\eend = t_l.\eend$. If $T_k$ and $T_l$ are read-only ($t_k.\updates = 
t_l.\updates = \emptyset$), the initial values of $t_k.\eend$ and $t_l.\eend$ 
do not change. Otherwise, both processes increment $\LC$ and assign its current 
value to the $\eend$ fields (line \ref{alg:htr:incSM}, the value of $\LC$ could 
not change during the execution of SM transactions). 

Because $t_k.\id = t_l.\id$, $t_k.\start = t_l.\start$,  $t_k.\readset = 
t_l.\readset$, $t_k.\updates = t_l.\updates$, and $t_k.\eend = t_l.\eend$, we 
gather that $t_k = t_l$ (except for the statistics field). If both 
transactions are updating, $p_i$ adds $t_k$ to $p_i$'s $\Log$ and $p_j$ adds 
$t_l$ to $p_j$'s $\Log$. Then, both processes apply all updates from the 
respective transaction descriptors. Thus both processes move to the same state, 
i.e., $S'$.
\end{proof}

\begin{proposition} \label{prop:htr:samesequence}
Let $S(i) = ( S^i_0, S^i_1, ... )$ be a sequence of states of a process $p_i$
running HTR, where $S^i_0$ is the initial state (comprising of the initial 
state of t-objects, $\LC = 0$ and $\Log = \emptyset$) and $S^i_k$ is the state 
after the $k$-th message was delivered using TOB and processed by $p_i$. For 
every pair of processes $p_i$ and $p_j$ either $S(i)$ is a prefix of $S(j)$ or 
$S(j)$ is a prefix of $S(i)$.
\end{proposition}

\begin{proof}
We prove the proposition by a contradiction. Let us assume that $S(i)$ and 
$S(j)$ differ on position $k$ (and $k$ is the lowest number for which 
$S^i_k \neq S^j_k$), thus neither is a prefix of another.

If $k = 0$, then the initial state of $p_i$ is different from the initial state
of $p_j$. Since all processes start with the same values of $\LC$, $\Log$ 
(lines \ref{alg:htr:lc}--\ref{alg:htr:log}) and maintain the same t-objects 
with the same initial values, that is a contradiction. Therefore $k > 0$, and 
the difference between $S^i_k$ and $S^j_k$ must stem from some later change to 
the state.


Since $S^i_{k-1} = S^j_{k-1}$, the receipt of the $k$-th message $m_k$ (which 
is equal for both processes thanks to the use of TOB) and processing it 
must have resulted in a different change to $\LC$, $\Log$ or values of some (or 
all) t-objects at both processes. We have two cases to consider:

\begin{enumerate}
\item Message $m_k$ is a transaction descriptor of a DU transaction (line 
\ref{alg:htr:adeliverTx}). Both processes are in the same state $S_{k-1}$ and 
process the same transaction descriptor. Therefore, by 
Proposition~\ref{prop:htr:deterministicdu}, both processes move to the same 
state $S_k$, a contradiction.

\item Message $m_k$ is a request to be executed as an SM transaction (line 
\ref{alg:htr:adeliverTxSM}). Both processes are in the same state $S_{k-1}$ and 
execute the same request as SM transactions. Therefore, by 
Proposition~\ref{prop:htr:deterministicsm}, both processes move to the same 
state $S_k$, a contradiction.
\end{enumerate}

Since both cases yield a contradiction, the assumption is false. Therefore 
either $S(i)$ is a prefix of $S(j)$, or $S(j)$ is a prefix of $S(i)$.

\end{proof}

%

We know that all communication between processes in HTR happens through TOB. It 
means that all processes deliver all messages in the same order. If the message 
is a request forwarded by some process to be executed as an SM transaction, 
then every process delivers this request while being in the same state (by 
Proposition~\ref{prop:htr:samesequence}). Then, all processes execute the 
request as different SM transactions but end up with transaction descriptors of 
the same exact value (except for the statistics field, by 
Proposition~\ref{prop:htr:deterministicsm}). Therefore, processes need not to
disseminate the transaction descriptors after they complete the transaction 
execution (as in case of a DU transaction). Instead, processes may promptly 
apply the updates from the transaction descriptors to their state. It all means 
that executing a request as multiple SM transactions across the whole system is 
equivalent to execution of the request only once and then distributing the 
resulting updates to all processes. Also, unless a request must be executed in 
the SM mode (because it performs some irrevocable operations), client has no 
knowledge which execution mode was chosen to execute his request.

The way HTR handles SM transactions means that HTR does not exactly fit the 
model of (update-real-time) opacity which requires that updates produced by 
every committed transaction must be accounted for. Therefore, unless 
the SM transactions resulting from execution of the same request did not 
perform any modifications or rolled back on demand, it is impossible to 
construct a t-sequential t-history $S$ in which every transaction is t-legal. 
However, since we proved that execution of multiple SM transactions regarding 
the same request is equivalent to execution of a single one, we can propose the 
following mapping of t-histories, which we call \emph{SMreduce}. Roughly 
speaking, under the \emph{SMreduce} mapping of some t-history of HTR, for any 
group of SM transactions regarding the same request $r$, such that the 
processes that executed the transactions applied the updates produced by the 
transactions, we allow only the first transaction of the group to commit in the 
t-history; other appear aborted in the transformed t-history.

Now let us give a formal definition of the \emph{SMreduce} mapping. Let $H$ 
be a t-history of HTR and let $\smmode$ be a predicate such that for any 
transaction $T_k$ in $H$, $\smmode(T_k)$ is true if $T_k$ was executed as a SM 
transaction in $H$. Otherwise $\smmode(T_k)$ is false. Then, let $H' = 
\mathit{SMreduce}(H)$ be a t-history constructed by changing $H$ in the 
following way. For any event $e$ in $H$ such that:
\begin{itemize}
\item $e = \resp_i(C_k)$ is a response event of an operation execution 
$M.\tryC(T_k) \rightarrow_i C_k$ for some transaction $T_k$ and process $p_i$,
and
\item $\smmode(T_k)$ is true (and $r$ is the request whose execution resulted 
in $T_k$), and
\item $T_k$ is not the first completed transaction in $H$ which
resulted from execution of $r$ in the SM mode,
\end{itemize}
replace $e$ in $H'$ with $e' = \resp_i(A_k)$. We say that $H'$ is an SMreduced 
t-history of HTR.

\begin{proposition} \label{prop:htr:samedesc}
Let $H$ by an SMreduced t-history of HTR. Let $T_k$ be an updating committed 
transaction in $H$ such that $t_k$ is the transaction descriptor of $T_k$. 
Then, any process replicates $t_k$ (excluding the statistics field) in its 
$\Log$ as $\LC$ on this process reaches $t_k.\eend$ (both actions happen 
atomically, i.e., in a lock statement).
\end{proposition}

\begin{proof}
From Proposition~\ref{prop:htr:samesequence} we know that all processes move 
through the same sequence of states as a result of delivering messages using 
TOB. Let $S$ be the state of any (correct) process immediately before 
delivering and processing a message $m$ such that processing of $m$ results in 
applying updates produced by $T_k$ to the system's state (we know that $T_k$ is 
an updating committed transaction, thus $t_k.\updates \neq \emptyset$). We have 
two cases to consider:
\begin{enumerate}
\item $T_k$ is a DU transaction. Then, by 
Proposition~\ref{prop:htr:deterministicdu}, upon delivery of $m$ any process 
$p_i$ updates its state in the same way and the new state includes a 
transaction descriptor $t'_k = t_k$ such that $t'_k$ is in $\Log$ of $p_i$. The 
value of $t'_k.\eend$ is equal to the current value of $\LC$ on $p_i$ because 
in the same lock statement $\LC$ is first incremented and then its current 
value is assigned to $t'_i.\eend$ (lines 
\ref{alg:htr:commitBegin}--\ref{alg:htr:commitEnd}).
\item $T_k$ is an SM transaction. Then, upon delivery of $m$ any process 
$p_i$ executes the request received in $m$ as an SM transaction $T'_k$ ($T'_k$ 
may or may not be equal $T_k$) with transaction descriptor $t'_k$. After $T'_k$ 
finishes its execution, inside the same lock statement $p_i$ increments the 
value of its $\LC$, appends $t'_k$ to its $\Log$ and applies updates produced 
by $T'_k$ (lines \ref{alg:htr:commitBeginSM}--\ref{alg:htr:commitEndSM}). By 
Proposition~\ref{prop:htr:deterministicsm}, $t'_k = t_k$ (except 
for the statistics field). Note that by definition of SMreduce, every 
transaction $T'_k \neq T_k$ such that $T'_k$ is executed as a result of receipt 
of $m$, $T'_k$ is aborted (which means that updates of a committed SM 
transaction are in fact applied only once by every process).
\end{enumerate}
This way for both cases we gather that $t_k$ is replicated in the $\Log$ on any 
process as $\LC$ on this process reaches $t_k.\eend$.
\end{proof}

\begin{proposition} \label{prop:htr:prec}
Let $H$ be an SMreduced t-history of HTR. For any two updating committed 
transactions $T_i, T_j \in H$ and their transaction descriptors $t_i$ and 
$t_j$, if $T_i \prec^r_H T_j$ then $t_i.\eend < t_j.\eend$.
\end{proposition}

\begin{proof}
From the assumption that $T_i \prec^r_H T_j$, we know that $T_i$ is committed 
and the first event of $T_j$ appears in $H$ after the last event of $T_i$ (the 
commit of $T_i$). It means that $\tryC(T_j)$ was invoked after the commit of 
$T_i$. Now we have four cases to consider:
\begin{enumerate}

\item $T_i$ and $T_j$ are DU transactions. Since $\tryC(T_j)$ is invoked after 
the commit of $T_i$, $t_j$ was broadcast (using TOB) after $t_i$ is delivered 
by the process that executed $T_i$. Hence, any process can deliver $t_j$ only 
after $t_i$. Since $\LC$ increases monotonically (line \ref{alg:htr:inc}) and 
its current value is assigned to the $\eend$ field of a transaction descriptor 
(line \ref{alg:htr:end}), on every process $t_i.\eend < t_j.\eend$. 

\item $T_i$ is a DU transaction and $T_j$ is an SM transaction (whose execution 
resulted from delivery of request $r$ using TOB). Since $T_j$ is committed, by 
definition of SMreduce, $T_j$ is the first SM transaction in $H$ to complete 
and such that $T_j$'s execution resulted from delivery of $r$. It means that 
there does not exist an SM transaction $T'_j$ on the process that executes 
$T_i$, whose execution resulted from delivery of request $r$ and which 
completed before $T_j$ did. Since $T_i$ commits before $T_j$, $T_i$ has to 
commit before any transaction $T''_j$ (whose execution also results from 
delivery of $r$) completes. By Proposition~\ref{prop:htr:noconcurrency}, 
$T_i$'s certification and commit and $T''_j$'s execution do not interleave. It 
means that $T''_j$ must have started after $T_i$ committed. Then $T_i$ must 
have incremented $\LC$ before $T''_j$ started (line \ref{alg:htr:inc}) and so 
$t_i.\eend < t''_j.\eend$, where $t''_j$ is the transaction descriptor of 
$T''_j$. By Proposition~\ref{prop:htr:deterministicsm}, $t''_j = t_j$. 
Therefore $t_i.\eend < t_j.\eend$.

\item $T_i$ is an SM transaction (whose execution resulted from delivery of 
request $r$ using TOB) and $T_j$ is a DU transaction. Since $\tryC(T_j)$ was 
invoked after the commit of $T_i$, $t_j$ was broadcast (using TOB) later than
$r$ was delivered by the process that executes $T_i$. Therefore, this process 
can only deliver $t_j$ after handling $r$ and executing $T_i$. By properties of
TOB, the process that executes $T_j$ can also deliver $t_j$ after delivery of 
$r$. Therefore, the process that executes $T_j$ had to deliver $r$, execute
an SM transaction $T'_i$ (with transaction descriptor $t'_i$) and modify the 
system's state afterwards but before $T_j$ started (by 
Proposition~\ref{prop:htr:noconcurrency}). Since $\LC$ increases monotonically 
(line \ref{alg:htr:incSM}), $t'_i.\eend < t_j.\eend$. By 
Proposition~\ref{prop:htr:deterministicsm}, $t'_i 
= t_i$, thus $t'_i.\eend = t_i.\eend$. Therefore $t_i.\eend < t_j.\eend$.

\item $T_i$ and $T_j$ are SM transactions (whose execution resulted from 
delivery using TOB of requests $r$ and $r'$, respectively). By properties of 
TOB, $r$ and $r'$ had to be delivered by any process in the same order. We now 
show that $r$ must be delivered by TOB before $r'$. Assume otherwise. Then, the 
process that executes $T_i$ delivers $r'$ prior to execution of $T_i$. As a 
result, the process executes an SM transaction $T'_j$ which produces the same 
results as $T_j$ (by Proposition~\ref{prop:htr:deterministicsm}). By 
Proposition~\ref{prop:htr:noconcurrency}, $T'_j$ must complete before $T_i$ 
starts. Then, $T'_j \prec^r_H T_i$. By definition of SMreduce, we know that 
$T_j$ is the first SM transaction to complete in $H$, such that $T_j$'s 
execution resulted from delivery of $r'$. Therefore $T_j$ must have completed 
before $T'_j$. It means that $T_j \prec^r_H T_i$, a contradiction. Therefore 
$r$ must be delivered by TOB prior to $r'$. Since the process that executes 
$T_j$ must have delivered $r$ before delivering $r'$, the process must have 
executed an SM transaction $T'_i$ prior to $T_j$ such that the execution of 
$T'_i$ resulted from delivery of $r$. By 
Proposition~\ref{prop:htr:noconcurrency}, $T'_i$ completes before $T_j$ starts.
Since $\LC$ increases monotonically (line \ref{alg:htr:incSM}), $t'_i.\eend < 
t_j.\eend$, where $t'_i$ is the transaction descriptor of $T'_i$. By 
Proposition~\ref{prop:htr:deterministicsm}, $t'_i = t_i$, thus $t'_i.\eend = 
t_i.\eend$. Therefore $t_i.\eend < t_j.\eend$.
\end{enumerate}
\end{proof}

\begin{proposition} \label{prop:htr:startend}
Let $H$ by an SMreduced t-history of HTR. Let $T_i$ and $T_j$ be two 
transactions in $H$ executed by some process $p_l$ and let $t_i$ and $t_j$ be 
transaction descriptors of $T_i$ and $T_j$, respectively. If $T_i \prec^r_H 
T_j$ then $t_i.\start \leq t_j.\start$.
\end{proposition}

\begin{proof}
From the assumption that $T_i \prec^r_H T_j$ and both are executed by the same 
process $p_l$, we know that $T_i$ is completed and the first event of $T_j$ 
appears in $H$ after the last event of $T_i$. Therefore $p_l$ assigns the 
current value of $\LC$ to $t_i.\start$ before it does so for $t_j.\start$ (in 
line \ref{alg:htr:start} and line~\ref{alg:htr:startSM}, if $T_i$ is a DU or SM 
transaction, respectively). The value of $\LC$ increases monotonically (lines 
\ref{alg:htr:inc} and \ref{alg:htr:incSM}, the values of $\LC$ correspond to 
commits of updating transactions). Therefore $t_i.\start \leq t_j.\start$.
\end{proof}

\begin{proposition} \label{prop:htr:readState}
Let $H$ be an SMreduced t-history of HTR and let $r = x.\rread \rightarrow v$ 
be a read operation on some t-object $x \in \mathcal{Q}$ performed by some 
transaction $T_k$ in $H$. If $T_k$ did not perform any write operations on $x$ 
prior to $r$ then either there exists a transaction $T_i$ that performed 
$x.\wwrite(v) \rightarrow \ok$ and committed before $r$ returns, or (if there 
is no such transaction $T_i$) $v$ is equal to the initial value of $x$.
\end{proposition}

\begin{proof}
From the assumption that $T_k$ did not perform any write operations on $x$ 
prior to $r$, we know that the value of $x$ is retrieved from the system state 
(line \ref{alg:htr:retrieve}). The value of $x$ is updated on the process that 
executes $T_k$ only in two cases:
\begin{enumerate}
\item A transaction descriptor $t'_i$ of an committed updating DU transaction 
$T'_i$ is delivered using TOB. Then $t'_i.\updates$ are used to modify $x$ in 
the system's state of the process that executes $T_k$ (line 
\ref{alg:htr:apply}). Before commit, $T'_i$ stores the modified values of 
t-objects in the $\updates$ set of the transaction descriptor of $T'_i$. The 
only possibility that a new value of $x$ is stored in the $\updates$ set is 
upon write operation on $x$ (line \ref{alg:htr:write}). Then $T_i = T'_i$ thus 
satisfying the Proposition.
\item An updating SM transaction $T'_i$ (whose execution resulted from delivery 
of a request $r_i$ using TOB) modified $x$ upon applying the updates it
produced (line \ref{alg:htr:applySM}); $T'_i$ is executed by the same process 
that executes $T_k$. Before that, during execution, $T'_i$ stores the modified 
values of t-objects in the $\updates$ sets of $T_i$'s transaction descriptor. 
The only possibility that a new value of $x$ is stored in the $\updates$ set is 
upon write operation on $x$ (line \ref{alg:htr:writeSM}). Now, because $H$ is 
SMreduced there are two cases to consider. In the first case $T'_i$ is 
committed. Then $T_i = T'_i$ thus satisfying the Proposition. In the second 
case $T'_i$ is aborted. However, from definition of SMreduce, we know that 
there exists a committed SM transaction $T_i$ whose execution resulted from 
delivery of $r_i$, such that $T_i$ committed before $T'_i$ completed (and 
therefore also prior to $r$). The transaction descriptor of $T_i$ is equivalent 
to the transaction descriptor of $T'_i$ (except for the statistics, by 
Proposition~\ref{prop:htr:deterministicsm}). Then, when $T_k$ performs $r$, the 
system's state contains the updates produced by $T_i$, i.e., it contains also 
$v$ as the current value of $x$.
\end{enumerate}

On the other hand, if the value of $x$ in the system was never updated (through 
line \ref{alg:htr:apply} or \ref{alg:htr:applySM}), the initial value of $x$ is 
returned (line \ref{alg:htr:retrieve}).
\end{proof}

\begin{proposition} \label{prop:htr:localWrite}
Let $H$ be an (SMreduced) t-history of HTR and $T_k$ (with a transaction 
descriptor $t_k$) be some transaction in $H$. If there exists a t-object $x \in 
\mathcal{Q}$, such that $T_k$ performs a read operation $r = x.\rread 
\rightarrow v$ and $T_k$ executed earlier at least one write operation on $x$, 
where $w = x.\wwrite(v') \rightarrow \ok$ is the last such an operation before 
$r$, then $v = v'$.
\end{proposition} 

\begin{proof}
Upon execution of $w$, if there were no prior write operations on $x$ in $T_k$ 
then a pair $(x, v')$ is added to $t_k.\updates$; otherwise, the current pair 
$(x, v")$ is substituted by $(x, v')$ in $t_k.\updates$ 
(line~\ref{alg:htr:write} or line~\ref{alg:htr:writeSM}, if $T_k$ is a DU or an 
SM transaction, respectively). Then $v'$ would be returned upon execution of 
$r$ (line~\ref{alg:htr:getObjectCall} and then \ref{alg:htr:retrieveUpdates}), 
unless $T_k$ aborts. This may only happen if $T_k$ is a DU transaction and 
$T_k$ aborted due to a conflict with other transaction 
(line~\ref{alg:htr:check}). However, then $r$ would not return any value. 
Therefore $v = v'$ and indeed $v'$ was assigned to $x$ by the last write 
operation on $x$ in $T_k$ before $r$.
\end{proof}


\htro*

\begin{proof}
In order to prove that HTR satisfies update-real-time opacity under SMreduce, 
we have to show that every SMreduced finite t-history produced by HTR is 
final-state update-real-time opaque (by 
Corollary~1 of \cite{KKW16}). In 
other words, we have to show that for every SMreduced finite t-history $H$ 
produced by HTR, there exists a t-sequential t-history $S$ equivalent to some 
completion of $H$, such that $S$ respects the update-real-time order of $H$ and 
every transaction $T_k$ in $S$ is legal in $S$.

\begin{proofpart}
Construction of a t-sequential t-history $S$ that is equivalent to a completion 
of $H$.

Let us first construct a t-completion $\comp{H}$ of $H$. We start with 
$\comp{H} = H$. Next, for each live transaction $T_k$ in $H$ performed by 
process $p_i$, we append some event to $\comp{H}$ according to the following 
rules:
\begin{itemize}
\item if $T_k$ is not commit-pending and the last event of $T_k$ is an 
invocation of some operation, append $\resp_i(A_k)$,
\item if $T_k$ is not commit-pending and the last event of $T_k$ is a 
response event to some operation, append $\langle \tryA(T_k) \rightarrow_i A_k 
\rangle$,
\item if $T_k$ is commit pending and $T_k$ is an SM transaction, then append 
$\resp_i(A_k)$,
\item if $T_k$ is commit pending and $T_k$ is a DU transaction with a 
transaction descriptor $t_k$, then if $t_k$ was delivered using TOB by some 
process $p_j$ and $p_j$ successfully certified $T_k$, then append 
$\resp_i(C_k)$, otherwise append $\resp_i(A_k)$.
\end{itemize}

Now we show that for each committed updating transaction $T_k$ there exists 
a unique value which corresponds to this transaction. This value is equal to 
the value of the $\eend$ field of $T_k$'s transaction descriptor when the 
updates of $T_k$ are applied (on any process), as we show by a contradiction. 
Let $T_i$ and $T_j$ be two updating committed transactions with transaction 
descriptors $t_i$ and $t_j$, respectively. $T_i$ and $T_j$ result from delivery 
of some requests $r_i$ and $r_j$ ($T_i$ and $T_j$ may be DU or SM 
transactions). Assume that $T_i \neq T_j$, $t_i.\eend = v$, $t_j.\eend = v'$, 
but $v = v'$. If $T_i$ is a DU transaction, then $t_i$ is broadcast 
using TOB to all processes in a message $m_i$ (line \ref{alg:htr:tob}). If 
$T_i$ is an SM transaction, then the request $r_i$ is broadcast using TOB prior 
to execution of $T_i$ in a message $m_i$ (line \ref{alg:htr:tobSM}). 
Analogically for $T_j$, message $m_j$ contains either $t_j$ or $r_j$. Since 
both transactions are updating committed, both $m_i$ and $m_j$ had to be 
delivered by some processes. By properties of TOB, we know that there exists a 
process $p$ that delivers both $m_i$ and $m_j$. Without loss of generality, let 
us assume that $p$ delivers $m_i$ before $m_j$. By 
Proposition~\ref{prop:htr:noconcurrency}, we know that processing of $m_i$ and 
$m_j$ cannot interleave (irrespective of the modes of the transactions) and the 
updates of $t_i$ and $t_j$ on $p$ are applied in the order of delivery of $m_i$ 
and $m_j$. Every time $p$ updates its state, $p$ first increments $\LC$ (line 
\ref{alg:htr:inc} or \ref{alg:htr:incSM}) and then assigns its value to the 
$\eend$ field of the currently processed transaction descriptor (line 
\ref{alg:htr:end} or \ref{alg:htr:endSM}). Therefore $v \neq v'$, a 
contradiction. Moreover, by Proposition~\ref{prop:htr:samesequence}, all 
processes deliver $m_i$ while being in the same state and then deliver $m_j$ 
while also being in the same state. Therefore the values of $\LC$ (and matching 
$\eend$ fields of transaction descriptors of updating committed transactions) 
are the same on every process when processing updates of $T_i$ and $T_j$. Thus, 
the value of the $\eend$ field of a transaction descriptor of a committed 
updating transaction uniquely identifies the transaction.

We can now construct the following function $\update$. Let $\update : 
\mathbb{N} \to \mathcal{T}$ be a function that maps the $\eend$ field of a 
transaction descriptor of a committed updating transaction to the transaction. 
Let $S = \langle \comp{H}|\update(1) \cdot \comp{H}|\update(2) \cdot ... 
\rangle$. This way $S$ includes the operations of all the committed updating 
transactions in $H$. Now, let us add the rest of transactions from $\comp{H}$ 
to $S$ in the following way. For every such a transaction $T_k$ with a 
transaction descriptor $t_k$, find a committed updating transaction $T_l$ (with 
transaction descriptor $t_l$) in $S$, such that $t_k.\start = t_l.\eend$, and 
insert $\comp{H}|T_k$ immediately after $T_l$'s operations in $S$. If there is 
no such transaction $T_l$ ($t_k.\start = 0$), then add $\comp{H}|T_k$ to the 
beginning of $S$. If there are multiple transactions with the same value of 
$\start$ timestamp, then insert them in the same place in $S$. Their relative 
order is irrelevant unless they are executed by the same process. In such a 
case, rearrange them in $S$ according to the order in which they were executed 
by the process.
\end{proofpart}

\begin{proofpart}
Proof that $S$ respects the update-real-time order of $H$.

Let $T_i$ and $T_j$ be any two transactions such that $T_i \prec^u_H T_j$ and 
let $t_i$ and $t_j$ be transaction descriptors of $T_i$ and $T_j$, 
respectively. Then, $T_i \prec^r_H T_j$ and 
\begin{enumerate}
\item $T_i$ and $T_j$ are updating and committed, or
\item $T_i$ and $T_j$ are executed by the same process.
\end{enumerate}

In case 1, by Proposition~\ref{prop:htr:prec} we know that $t_i.\eend < 
t_j.\eend$. Both $t_i.\eend$ and $t_j.\eend$ correspond to the 
values assigned to $\LC$ when $t_i$ and $t_j$ are processed (lines 
\ref{alg:htr:inc} and \ref{alg:htr:end}). Then, by the construction of $S$, 
$T_i$ must appear in $S$ before $T_j$. Therefore, $T_i \prec^r_S T_j$. 
Moreover, the construction requires that for any transaction $T_k \in H$, $S$ 
includes all events of $H|T_k$. In turn both $T_i$ and $T_j$ are updating and 
committed in $S$. Therefore, in this case, $T_i \prec^u_S T_j$.

Now let us consider case 2. We have several subcases to consider:
\begin{enumerate}
\item $T_i$ is a committed updating transaction and $T_j$ is a read-only or 
an aborted transaction. Since $T_i \prec^r_H T_j$ and both $T_i$ and $T_j$ are 
executed by the same process, naturally $t_i.\eend \leq t_j.\start$ (the value 
of $\LC$, which is assigned to the $\start$ and $\eend$ fields of transaction 
descriptor, increases monotonically). By construction of $S$, $T_j$ (which is a 
read-only or an aborted transaction) appears in $S$ after a committed updating 
transaction $T_k$ (with transaction descriptor $t_k$), such that $t_k.\eend = 
t_j.\start$. Therefore $T_k \prec^r_S T_j$ and $t_i.\eend \leq t_k.\eend$. If 
$t_i.\eend = t_k.\eend$, then $T_i = T_k$ and $T_i \prec^r_S T_j$. If 
$t_i.\eend < t_k.\eend$, by construction of $S$, $T_i \prec^r_S T_k$, and thus 
$T_i \prec^r_S T_j$.

\item $T_i$ is a read-only or an aborted transaction and $T_j$ is a committed 
updating transaction. By Proposition~\ref{prop:htr:startend}, $t_i.\start \leq 
t_j.\start$. Since $T_j$ is a committed updating transaction, $t_j.\start < 
t_j.\eend$ ($\LC$ is always incremented prior to assigning it to 
the $\eend$ field of the transaction descriptor upon transaction commit). 
Therefore $t_i.\start \leq t_j.\start < t_j.\eend$, and thus $t_i.\start < 
t_j.\eend$. Since $T_i$ is a read-only or an aborted transaction, by 
construction of $S$, $T_i$ appears in $S$ after some committed updating 
transaction $T_k$ (with transaction descriptor $t_k$) such that $t_k.\eend = 
t_i.\start$ and before some committed updating transaction $T'_k$ (with 
transaction descriptor $t'_k$) such that $t'_k.\eend = t_k.\eend + 1$. $T_k$ 
may exist or may not exist. We consider both cases:
\begin{enumerate}
\item $T_k$ exists. It means that $T_k \prec^r_S T_i \prec^r_S T'_k$. Since 
$t_i.\start < t_j.\eend$ and $t_k.\eend = t_i.\start$, $t_k.\eend < t_j.\eend$. 
Because $t_k.\eend + 1 = t'_k.\eend$, $t'_k.\eend \leq t_j.\eend$. If 
$t'_k.\eend = t_j.\eend$, then $T'_k = T_j$ and $T_i \prec^r_S T_j$. If 
$t'_k.\eend < t_j.\eend$, by construction of $S$, $T'_k \prec^r_S T_j$, and 
thus $T_i \prec^r_S T_j$.

\item $T_k$ does not exist. It means that there is no committed updating 
transaction in $S$ before $T_i$ ($t_i.\start = 0$). By construction of $S$, 
$T_i$ is placed at the beginning of $S$, before any committed updating 
transaction. Therefore $T_i \prec^r_S T_j$.
\end{enumerate}

\item Both $T_i$ and $T_j$ are read-only or aborted transactions. From 
Proposition~\ref{prop:htr:startend} we know that $t_i.\start \leq t_j.\start$.
If $t_i.\start = t_j.\start$ (and both $T_i$ and $T_j$ are executed 
by the same process), then the construction of $S$ explicitly requires that 
$T_i$ and $T_j$ are ordered in $S$ according to the order in which they were 
executed by this process. On the other hand, if $t_i.\start < t_j.\start$ then 
by the construction of $S$:
\begin{enumerate}
\item $T_i$ and $T_j$ appear in $S$ after some committed updating 
transactions 
$T'_i$ and $T'_j$ with transaction descriptors $t'_i$ and $t'_j$ such that 
$t_i.\start = t'_i.\eend$ and $t_j.\start = t'_j.\eend$. It means that 
$t'_i.\eend < t'_j.\eend$, therefore $T'_i$ appears in $S$ before $T'_j$ (by 
the construction of $S$). Moreover, between $T'_i$ and $T_i$ in $S$ there is 
no other committed updating transaction, since, by the construction of $S$, 
$T_i$ is inserted immediately after $T'_i$. In turn, the four transactions 
appear in $S$ in the following order: $T'_i$, $T_i$, $T'_j$, $T_j$. Thus 
$T_i \prec^r_S T_j$.
\item If such $T'_i$ does not exist ($t_i.\start = 0$; there is no committed
updating transaction in $S$ before $T_i$), we know that $T'_j$ has to exist 
since $t'_j.\start = t_j.\eend > t_i.\start = 0$. Then, the three transactions 
appear in $S$ in the following order: $T_i$, $T'_j$, $T_j$. Thus also $T_i 
\prec^r_S T_j$.
\end{enumerate}
\end{enumerate}

This way $S$ respects the update-real-time order of $H$ (trivially, for any 
transaction $T_k$ executed by process $p_i$ in $H$, $T_k$ is executed by $p_i$ 
in $S$).
\end{proofpart}

\begin{proofpart}
Proof that every transaction $T_j$ in $S$ is legal in $S$.

We give the proof by contradiction. Assume that there exists a transaction 
$T_j$ (with a transaction descriptor $t_j$ and executed by some process $p$) 
such that $T_j$ is the first transaction that is not legal in $S$. It means 
that there exists $x \in \mathcal{Q}$ such that $\vis = \visible_S(T_j)|x$ does 
not satisfy the sequential specification of $x$.

The only type of t-object considered in HTR are simple registers. 
Sequential specification of a register $x$ is 
violated when a \emph{read} operation $r = x.\rread \rightarrow v$ 
returns a value $v$ that is different from the most recently written value to 
this register using the \emph{write} operation, or its initial value if there 
was no such operation. 

Therefore, $\vis$ does not satisfy the sequential specification of $x$, if 
there exists an operation $r = x.\rread \rightarrow v$ in $T_j$ such that $v$ 
is not the most recently written value to $x$ in $\vis$. Then, either $v' \neq 
v$ is the initial value of $x$ or there exists an operation $w = x.\wwrite(v') 
\rightarrow \ok$ in $\vis$ such that $w$ is the most recent write operation on 
$x$ in $\vis$ prior to $r$.

By definition of $\visible_S(T_j)$, instead of considering t-history $\vis$, we 
can simply operate on $S$ while excluding from consideration any write 
operations performed by all aborted transactions in $S$.

Let us first assume that $x$ was not modified prior to $r$, i.e., there is no 
write operation execution on $x$ in $S$ (and in $\vis$) prior to $r$. Then, 
trivially, $v$ has to be equal to the initial value of $x$ (by 
Proposition~\ref{prop:htr:readState}), a contradiction.

Therefore, there exists a transaction $T_i$ (with transaction descriptor $t_i$) 
which executes $w$. First, assume that $T_i = T_j$. Given that $w$ is executed 
prior to $r$, from Proposition~\ref{prop:htr:localWrite}, $v = v'$, a 
contradiction. Therefore $T_i \neq T_j$.

Since we require that $w$ is in $\vis$, $T_i$ must be a committed updating 
transaction and $T_i \prec^r_S T_j$. 

Now we show that $t_i.\eend \le t_j.\start$. We have two cases to consider:
\begin{enumerate}

\item $T_j$ is an aborted or read-only transaction in $S$. By construction of 
$S$, there exists a committed updating transaction $T_l$ (with transaction 
descriptor $t_l$) such that $T_l \prec^r_S T_j$ and $t_l.\eend = t_j.\start$. 
Because both $T_i$ and $T_l$ are committed updating transactions in $S$, either 
$T_l \prec^r_S T_i$, $T_i \prec^r_S T_l$ or $T_i = T_l$. By construction of 
$S$, between $T_l$ and $T_j$ there must be no committed updating transactions. 
If $T_l \prec^r_S T_i$ then $T_i$ must appear after $T_j$ in $S$. However, it 
is impossible since $T_i \prec^r_S T_j$, a contradiction. Then, either $T_i = 
T_l$ or $T_i \prec^r_S T_l$. In the first case, $t_i.\eend = t_l.\eend$. In the 
second case, $t_i.\eend < t_l.\eend$ (by Proposition~\ref{prop:htr:prec}). 
Since $t_l.\eend = t_j.\start$, $t_i.\eend \le t_j.\start$.

\item $T_j$ is a committed updating transaction in $S$. By 
Proposition~\ref{prop:htr:prec}, $t_i.\eend < t_j.\eend$. Now we have 
additional two cases to consider:
\begin{enumerate}
\item $T_j$ is a DU transaction. By Proposition~\ref{prop:htr:samedesc}, we 
know that $t_i$ is replicated in $\Log$ of $p$ (process that executes $T_j$) by 
the time the value of $\LC$ on that process reaches $t_i.\eend$. Since $T_j$ is 
a committed updating transaction, it has to pass the certification test (line 
\ref{alg:htr:certify}). This test takes place as late as the commit of $T_j$ 
(line \ref{alg:htr:globalCert}). Since the commit sets the value of $\LC$ to 
$t_j.\eend$ (line \ref{alg:htr:end}), the certification takes place when $\LC = 
t_j.\eend - 1$. Since $t_i.\eend < t_j.\eend$, $t_i.\eend \le t_j.\eend - 1$. 
This means that $t_i$ is already replicated in the $\Log$ of $p$ when 
certification happens. We know that $x \in t_j.\readset$ and $(x, v') \in 
t_i.\updates$. If $t_i.\eend > t_j.\start$, then the certification procedure 
would compare the $T_j$'s readset against the $T_i$'s updates and return 
$\mathit{failure}$, thus aborting $T_j$. But we know that $T_j$ is committed. 
Therefore, $t_i.\eend \le t_j.\start$. 
\item $T_j$ is an SM transaction. For any committed updating SM transaction $T$ 
(with transaction descriptor $t$) the following holds: $t.\start + 1 = t.\eend$
(by Proposition~\ref{prop:htr:noconcurrency}, execution of $T$ cannot 
interleave with execution of another SM transaction or handling of delivery of 
a transaction descriptor of a DU transaction). Since $t_i.\eend < t_j.\eend$ we 
know that $t_i.\eend < t_j.\start + 1$. Thus $t_i.\eend \le t_j.\start$.
\end{enumerate}

\end{enumerate}

By Proposition~\ref{prop:htr:samedesc} and the fact that $t_i.\eend \le 
t_j.\start$, we know that $t_i$ is replicated in $\Log$ of $p$ (process that 
executes $T_j$) before $T_j$ starts. It means that inside the same lock 
statement, $\LC$ is incremented (lines \ref{alg:htr:inc} and 
\ref{alg:htr:incSM}) and its value is assigned to $t_i.\eend$ (lines 
\ref{alg:htr:end} and \ref{alg:htr:endSM}), $t_i$ is appended to $\Log$ (lines 
\ref{alg:htr:logAppend} and \ref{alg:htr:logAppendSM}) and $t_i.\updates$ are 
applied to the system state (lines \ref{alg:htr:apply} and 
\ref{alg:htr:applySM}). Therefore, the updates of $T_i$ are applied to the 
system state of $p$ before $T_j$ starts.

Now, unless there is some transaction $T_k$ (with transaction descriptor 
$t_k$), such that $T_k$ modified $x$, $t_k.\updates$ are applied to the 
system's state of $p$ after $t_i.\updates$ are applied but before $r$ returns, 
$r$ would have to return $v'$. It is impossible, because we assumed that 
$r$ returns $v \neq v'$. Therefore we now consider such $T_k$. We have two 
cases to consider:
\begin{enumerate}
\item $T_k$ is a committed updating DU transaction or a committed updating SM 
transaction executed by $p$. If $t_k.\updates$ are indeed applied by $p$ after 
$t_i.\updates$ are, then $t_k.\eend > t_i.\eend$ ($p$ increments $\LC$ each 
time $p$ applies updates of some transaction, line \ref{alg:htr:inc} or 
\ref{alg:htr:incSM}). By construction of $S$, $T_k$ would have to appear in $S$ 
after $T_i$ and before $r$ returns. However, then $w$ would not be the most 
recent write operation on $x$ prior to $r$ in $S$, a contradiction.
\item $T_k$ is an aborted SM transaction executed by $p$, such that $T_k$'s 
execution resulted from delivery using TOB of some request $r_k$. Since $p$ 
applies $t_k.\updates$ after $t_i.\updates$, $t_k.\eend > t_i.\eend$ (lines 
\ref{alg:htr:inc} and \ref{alg:htr:incSM}). By definition of SMreduce, the 
updates of $T_k$ are applied to the system's state of $p$ only if there exists 
a committed updating transaction $T'_k$ (with transaction descriptor $t'_k$) 
whose execution also resulted from delivery of $r_k$. By 
Proposition~\ref{prop:htr:deterministicsm}, $t_k = t'_k$, and thus 
$t_k.\eend = t'_k.\eend$. Hence, $t'_k.\eend > t_i.\eend$. By construction of 
$S$, it means that $T'_k$ appears in $S$ after $T_i$ and before $r$ returns. 
However, then $w$ would not be the most recent write operation on $x$ prior to 
$r$ in $S$, a contradiction.
\end{enumerate}
Since both cases yield contradiction, the assumption that there exists such 
transaction $T_k$ is false. Therefore $r$ has to return $v = v'$ thus 
concluding the proof by contradiction. Therefore HTR guarantees 
update-real-time opacity under SMreduce.
\qedhere
\end{proofpart}
\end{proof}